\documentclass[a4paper,11pt]{article}
\usepackage[pub]{zyanres}
\usepackage{bm}
\DeclareSymbolFont{matha}{OML}{txmi}{m}{it}
\DeclareMathSymbol{v}{\mathord}{matha}{118}
\numberwithin{equation}{section}
\usepackage{tablefootnote}

\title{Linearised Second Law for Higher Curvature Gravity and Non-Minimally Coupled Vector Fields}

\author{Aron~C.~Wall\footnote{aroncwall@gmail.com} and Zihan~Yan\footnote{zy286@cam.ac.uk}\\~\\
\textit{DAMTP, Centre for Mathematical Sciences, University of Cambridge\\Wilberforce Road, Cambridge, U.K. CB3 0WA}}
\date{\today}

\begin{document}
    \maketitle

    \vspace{20pt}    
    \begin{abstract}
    Expanding the work of \cite{Wall:2015raa}, we show that black holes obey a second law for linear perturbations to bifurcate Killing horizons, in any covariant higher curvature gravity coupled to scalar and vector fields.  The vector fields do not need to be gauged, and (like the scalars) can have arbitrary non-minimal couplings to the metric.  The increasing entropy has a natural expression in covariant phase space language, which makes it manifestly invariant under JKM ambiguities.  An explicit entropy formula is given for $f(\text{Riemann})$ gravity coupled to vectors, where at most one derivative acts on each vector.  Besides the previously known curvature terms, there are three extra terms involving differentiating the Lagrangian by the symmetric vector derivative (which therefore vanish for gauge fields).
    \end{abstract}

    \newpage
    \tableofcontents
    \section{Introduction}

    In General Relativity, black holes have an entropy that is proportional to the area of their horizon:
    \begin{equation}\label{BH}
        S_{\text{BH}} = \frac{A}{4 G \hbar}
    \end{equation} 
    (Henceforth we set $\hbar = 1$, even though our analysis will be purely classical.) 
    Although quantum-mechanical arguments are necessary to fix the proportionality constant, in classical General Relativity one can check that the area $A$ is increasing with time in any dynamical process \cite{Hawking:1971tu}.  Leaving aside some technical details, the core reason this is true is as follows: If we define the expansion $\theta$ as the rate of change of area element $\delta A$ in the null direction $v$:
    \begin{equation}
        \theta = \frac{1}{\delta A} \dv{(\delta A)}{v},
    \end{equation}
    then light rays are always focussing by the Raychaudhuri equation:
    \begin{equation}\label{Raych}
        \dv{\theta}{v} \,=\, -\frac{\theta^2}{D-2} - \sigma_{ij}\sigma^{ij} - R_{vv} \:\le\: 0,
    \end{equation}
    assuming that matter obeys the null energy condition so that $R_{vv} = 8\pi GT_{vv} \ge 0$. 

    The LHS of \eqref{Raych} involves a \emph{second} $v$ derivative of the area.  However, for a future horizon at late times, it is possible to impose the boundary condition $\theta|_{v = +\infty} = 0$.  It is therefore possible to integrate \eqref{Raych} along the $v$ direction to show that $\theta \ge 0$, i.e., the area is increasing.  In black hole thermodynamics, this can be interpreted as a manifestation of the Second Law of Thermodynamics.

    The argument above uses the Einstein equation.  But there are good reasons to believe that there should be corrections to the Einstein-Hilbert action due to quantum loop corrections, stringy physics, and/or other sources.  It is therefore an important question whether the Second Law continues to hold, if the action is modified by, e.g., generic higher curvature corrections. For any higher curvature gravity theory, it is possible to construct a formula for increasing black hole entropy $S_\text{inc}$ which increases with time in a certain regime.  So far this has been done generally only for linearised perturbations to Killing horizons \cite{Wall:2015raa,Bhattacharyya:2021jhr} (or at higher order, within the realm of effective field theory where the higher curvature corrections are assumed to have perturbatively small effects \cite{Hollands:2022fkn}).

    The first goal of this paper is to clarify the relationship between this increasing entropy $S_\text{inc}$ and the Noether charge formalism of Wald and Iyer \cite{Wald:1993nt,Iyer:1994ys}.  In particular, the Noether charge entropy --- obtained by integrating the boundary charge ${\bm Q}_{\xi}$ over the horizon, with $\xi$ a local boost generator on the horizon --- is subject to certain ``JKM ambiguities'' in the case of dynamically evolving black holes \cite{Jacobson:1993vj,Iyer:1994ys}.  On the other hand $S_\text{inc}$ is unambiguous at linear order (on a compact horizon), and was constructed in a seemingly unrelated way, by taking the linearly perturbed null-null equation of motion $\delta E_{vv}$ and pulling two $v$-derivatives out of it to obtain $\delta E_{vv} = (\partial_v)^2 \delta S_\text{inc}$. (Throughout the paper, $\delta$ labels the linear perturbation to the background.)

    However, it turns out that that there is a useful covariant relation between the second derivative of $\delta S_\text{inc}$ and a different Noether current associated with the affine null translation vector $k$ on the horizon, namely $(\pounds_k)^2 \delta S_\text{inc} = 2 \pi \int \delta (\pounds_k \bm Q_k - \iota_k \bm \Theta_k)$, where $\bm \Theta_k$ is the pre-symplectic potential, $\iota_k$ is the contraction with $k$, and $\pounds_k$ is the Lie derivative with respect to $k$. This expression is unambiguous because the JKM ambiguity cancels between the two terms. 

    Our second aim is to show that black hole entropy is well-defined in the presence of vector fields $V^a$, even in the absence of any gauge-symmetry for the vector field.  Although our general analysis is valid for an arbitrarily number of derivatives, we give an explicit formula $\eqref{vectorS}$ for $S_\text{inc}$, in the special case where the action depends on only one derivative of the vector $V^{a}$.  There are some new terms in $S_\text{inc}$ (beyond the previously known terms that involve differentiating by Riemann) if the action depends on the symmetrised derivative $\nabla_{(a} V_{b)}$.  (These new terms do not contribute for a gauge field $A_a$ where the action only depends on $F_{ab} = 2\nabla_{[a} A_{b]}$, and thus do not appear in the recent work \cite{Biswas:2022grc} on black hole entropy for gauge fields.)
    
    \subsection{The Noether Charge Method and Its Ambiguities}\label{Noether_review}

    In a general gravity theory with higher curvature and/or non-minimal curvature couplings, it is known that the formula for black hole entropy $S$ receives corrections, which are proportional to the coefficients of various curvature terms in the action.  In 1993, Wald \cite{Wald:1993nt} defined an entropy formula (Wald entropy) for \emph{stationary} black holes in general gravitational theories using the covariant phase space formalism, in which the entropy is identified as the \emph{Noether charge} $\bm Q_\xi$ associated with one of the Killing vectors $\xi$ (the one which looks like a boost near the horizon):
    \begin{equation}
        S_{\text{Noether}} = \frac{2 \pi}{\kappa} \int_{\mathcal{B}} \bm Q_\xi \simeq - 2 \pi \int_{\mathcal{B}} \dd[D-2]{x} \sqrt{h}\, \varepsilon_{ab} \varepsilon_{cd} E^{abcd}_R = S_{\text{Wald}},
    \end{equation}
    where the middle equality holds only if we make a particular choice of the JKM ambiguity \cite{Jacobson:1993vj,Iyer:1994ys}.  Here, $\mathcal{B}$ is the bifurcation surface with intrinsic metric $h_{ab}$, and $\varepsilon_{ab}$ is the binormal to $\mathcal{B}$, and 
    \begin{equation}
        E^{abcd}_R = \pdv{L}{R_{abcd}} - \nabla_{e_1} \pdv{L}{(\nabla_{e_1} R_{abcd})}+ \cdots + (-1)^p \nabla_{(e_1}\cdots \nabla_{e_p)} \pdv{L}{(\nabla_{(e_1} \cdots \nabla_{e_p)} R_{abcd})}.
    \end{equation} 
    In the case of General Relativity, $S_\text{Wald}$ reduces to the area \eqref{BH}.
 
    Regardless of how we fix the ambiguity, any Noether charge entropy readily satisfies a \emph{First Law} for linear perturbations to a Killing horizon:
    \begin{equation}
        \frac{\kappa}{2 \pi}\delta S_{\text{Noether}} = \delta M - \Omega_H \delta J - \Phi_H \delta Q
    \end{equation}
    where $\kappa$ is the surface gravity, $M$ is the mass, $J$ is the angular momentum and $Q$ is the charge of the black hole \cite{Wald:1993nt}.

    Unfortunately, for dynamically evolving black holes, the Noether charge formulation is subject to ambiguities, and $S_\text{Wald}$ is merely one possible way to fix the ambiguities (and it turns out, not the correct choice to obtain a Second Law in a general theory).  The most important type of JKM ambiguity\footnote{More precisely, JKM identified three different types of ambiguities affecting $\bm Q_\xi$, which will be labelled in section \ref{sec:jkm} as $\bm \alpha, \bm \beta, \bm\gamma$.  But the first ambiguity $\bm \alpha$, which is associated with adding a total derivative to the Lagrangian, does not affect terms in $\bm Q_\xi$ involving derivatives of $\xi$, so it does not matter when we consider a local boost like \eqref{locboost}.  And the third ambiguity $\bm \gamma$ is a total divergence on the horizon, which vanishes whenever the horizon is compact.  Hence, only the second type of JKM ambiguity, the $\bm \beta$ term, will matter to us.  This ambiguity shifts $\bm Q_\xi$ by an arbitrary Lie derivative w.r.t.~$\xi$, and if we restrict to boost-invariant expressions (as required by covariance in the $uv$ plane), that implies that the ambiguity takes the form given by \eqref{JKMproduct}.  In section \ref{sec:jkm} we will adopt a different point of view, in which the entropy is obtained from the horizon null translation vector in a way that is unambiguous.} takes the following form:  Suppose we wish to evaluate $S_\text{Wald}$ on a horizon slice ${\cal C}(v')$ located at some null time $v'$.  By diffeomorphism invariance, the entropy $S$ is required to be symmetric under a local Lorentz boost $\xi'$ of the form:
    \begin{equation}\label{locboost}
        \xi' = (v-v')\partial_v - u \partial_u
    \end{equation}
    pivoting around the slice ${\cal C}(v')$.  Suppose we now calculate the Noether charge at $v = v'$ with respect to this local boost vector.  However, there exist Lorentz invariant terms which are sums of products of the form:
    \begin{equation}\label{JKMproduct}
        S_\text{JKM} = \sum A_{(m)}B_{(-m)},\qquad m > 0,
    \end{equation}
    where $A_{(m)}$ has positive boost weight and $B_{(-m)}$ has negative boost weight.  On a bifurcate Killing spacetime, symmetry and regularity require that $A_{(m)} = 0$ everywhere on the future horizon, while in general $B_{(-m)} \ne 0$ (away from the bifurcation surface $\mathcal{B}$ at $v = 0$).  Hence, even for a linearised perturbation to a Killing horizon, we can have $\delta A_{(m)} \ne 0$ and so these ambiguities start to matter.  Without a method to fix these ambiguities, the analysis of \cite{Iyer:1994ys} is insufficient to determine which of these possible expressions for entropy is correct.  In general, most ways of fixing the ambiguity do not lead to an increasing quantity.

    \subsection{A Linearised Second Law}    

    Recently, significant progress has been made in showing that a Second Law holds quite generally for \emph{linearised} perturbations to Killing horizons, for a completely generic metric-scalar theory whose Lagrangian $L$ can be expanded in an arbitrarily high number of derivatives \cite{Kolekar:2012tq,Sarkar:2013swa,Bhattacharjee:2015yaa, Wall:2015raa}:

    \begin{equation}
        I = \int \dd[D]{x} \sqrt{-g}\, L(g^{ab}, R_{abcd}, \nabla_{e_1} R_{abcd}, \cdots, \nabla_{(e_1} \cdots \nabla_{e_p)} R_{abcd}, \phi, \nabla_{a_1} \phi, \cdots, \nabla_{(a_1} \cdots \nabla_{a_q)}\phi).
    \end{equation}
    Here, $L$ is constructed from all possible contractions of the field contents: the (inverse) metric $g^{ab}$, the Riemann tensor $R_{abcd}$, some scalar field $\phi$ and their symmetrised covariant derivatives.\footnote{The restriction to symmetric derivatives is without loss of generality as non-symmetric covariant derivatives can be rewritten using the Ricci identity.}
    
    In 2015, Wall \cite{Wall:2015raa} showed how to fix the JKM ambiguities for the full Noether charge, at the linear order of perturbation, by using the null-null component of the equation of motion on the horizon.  The increasing entropy $S_\text{inc}$ calculated there satisfies a \emph{linearised Second Law}, i.e., $\partial_v S_{\text{inc}} = 0$ for arbitrary linear perturbations of the gravitational fields, and $\partial_v S_{\text{inc}} \geq 0$ in the presence of an external matter source satisfying the null energy condition. The entropy is defined to linear order by
    \begin{equation}\label{Sinc}
        \partial_v^2 S_{\text{inc}} = - 2 \pi \int_{\mathcal{C}(v)} \dd[D-2]{x} \sqrt{h}\, E_{vv}
    \end{equation}   
    where Gaussian null coordinates $\{u,v,x^i\}$ (for a detailed discussion, see Section \ref{sec:gnc}) were chosen to carry out the calculation. Here, $v$ is the affine null time on the horizon, $u$ affinely parametrises the other null direction, and $\{x^i\}$ are $D-2$ spatial coordinates. The coordinates are chosen such that $u=0$ labels the horizon $\mathcal H$. $\mathcal{C}(v)$ is a time-slice of the horizon at $v$, $\sqrt{h}$ is the codimension-2 volume factor, and $E_{ab}$ is the equation of motion of the (inverse) metric $g^{ab}$. These notations are used throughout this article.
    
    For the stationary background, $E_{vv} = 0$ since it has Killing weight 2 (details shown in Section \ref{sec:gnc}), and similarly $\partial_v S_\text{inc} = 0$. Hence, at zeroth order \eqref{Sinc} is trivially satisfied. At first order, the content of the equation is:
    \begin{equation}\label{pulloff}
        \partial_v^2\, \delta S_{\text{inc}} = - 2 \pi \int_{\mathcal{C}(v)} \dd[D-2]{x} \delta \!\left(\sqrt{h}\, E_{vv} \right).
    \end{equation}
    By an algebraic argument using the gauge choice, \cite{Wall:2015raa} showed that it is always possible to extract two $\partial_v$ derivatives from $\delta(\sqrt{h}\, E_{vv})$ and thus obtain an entropy variation $\delta S_{\text{inc}}$.\footnote{This was subject to a caveat about exactness of the zero-boost term, whose resolution will be described in the next subsection.}
    
   \eqref{pulloff} holds off-shell, but we can also make an on-shell perturbation of order $\epsilon$, $\delta E_{vv} = \delta T_{vv} \sim \mathcal{O}(\epsilon)$ if we turn on some \emph{external} energy-momentum tensor $\delta T_{vv}$ sourcing the perturbation. The linearised second law $\partial_v S_\text{inc} \geq 0$ then holds if we assume null energy condition (NEC) for $\delta T_{vv}$ and that the perturbation settles down at late time so that $\partial_v S_\text{inc}|_{(v=+\infty)} = 0$.

    If we relax the assumption of an external matter source, we have $\delta E_{vv} = 0$ on shell, so the linearised second law actually says the entropy is constant.  Another way to see this is that, at linear order, we can always reverse the sign of perturbation parameter $\epsilon$ so that the entropy is both non-decreasing and non-increasing, hence it is constant.

    In the special case of $f(\text{Riemann})$ theories, with $L = L(g^{ab}, R_{abcd})$, the increasing entropy was calculated more explicitly in \cite{Wall:2015raa} as
    \begin{equation}
        S_{\text{inc}} = - 2 \pi \int_{\mathcal{C}(v)} \dd[D-2]{x} \sqrt{h} \left(4 \pdv{L}{R_{uvuv}} + 16 \pdv{L}{R_{uiuj}}{R_{vkvl}} \bar K_{ij} K_{kl}\right) + \mathcal{O}(\epsilon^2) \label{eq:wall-entropy}
    \end{equation}
    where $\bar K_{ij}$ and $K_{ij}$ are the extrinsic curvatures in $u$ and $v$ directions, respectively.

    The above proposal is proven to be valid for pure gravity and scalar fields (non-minimally) coupled to gravity. The validity for \emph{gauge fields} is argued in \cite{Wall:2015raa}, but was more carefully proven by Biswas et al.~\cite{Biswas:2022grc}. A recent paper \cite{Deo:2023vvb} by Deo et al.~also extended the validity to the case of Chern-Simons theories of gravity.

    \subsection{Role of the Entropy Current}  
    As stated above, \cite{Wall:2015raa} extracted two $\partial_v$ derivatives from $\delta(\sqrt{h}\, E_{vv})$ in order to obtain an entropy variation $^{\boldsymbol{-}}\!\!\!\!\delta S_{\text{inc}}$, where we now put the bar in the delta to indicate that this expression was not yet manifestly an exact variation.  Furthermore, this variation can be decomposed into a JKM ambiguity which \emph{is} automatically an exact variation, plus an Iyer-Wald ``zero-boost weight term'' which was not manifestly exact in the field variation: $^{\boldsymbol{-}}\!\!\!\!\delta S_{\text{inc}} = \delta S_{\text{JKM}} +\, ^{\boldsymbol{-}}\!\!\!\!\delta S_{\text{ZB}}$.  In order to argue that the latter term is exact (and hence that $^{\boldsymbol{-}}\!\!\!\!\delta S_{\text{inc}} = \delta ( S_\text{inc})$ for some geometrical $S_\text{inc}$ defined without reference to the field variation $\delta g_{ab}$), \cite{Wall:2015raa} appealed to the fact that the Iyer-Wald entropy satisfies the First Law, and can thus be substituted in for the $\delta S_{\text{ZB}}$ piece of the entropy. 

    Another related problem is that the above treatment only considered an integrated formula over some horizon slice $\mathcal{C}(v)$ with $v = \text{const}$.\footnote{Although any such slice may be put into this form by means of a suitable gauge choice, in \cite{Wall:2015raa} the gauge-invariance was not manifest.}

    This unsatisfactory situation has been resolved by work uncovering the important role of the entropy current, allowing $\delta S_{\text{ZB}}$ to be explicitly constructed as an exact variation.  To investigate what happens locally, in 2021, Bhattacharyya et al.~\cite{Bhattacharyya:2021jhr}  (see also \cite{Bhattacharya:2019qal} for a previous analysis of four-derivative theories) showed the off-shell structure of the $vv$-component of equation of motion is
    \begin{equation}
        E_{vv} = \frac{1}{\sqrt{h}} \partial_v \left(\sqrt{h}\left(\partial_v\varsigma + D_i J^i\right)\right) + \mathcal{O}(\epsilon^2)
    \end{equation}
    where $\varsigma$ is the \emph{entropy density}, $J^i$ is the \emph{entropy current}, and $D_i$ is the codimension-2 covariant derivative along the horizon slice.  (Note that the current $J^i$ vanishes for General Relativity.)  This need for a spatial boundary term for Gauss-Bonnet gravity had been previously anticipated in \cite{Guedens:2011dy} where --- in the context of a higher curvature equation of state argument --- it was found necessary for the entropy to be defined along $D-2$ dimensional slices of the horizon sharing a $D-3$ dimensional boundary, which effectively prevents entropy from escaping from the horizon region spatially, by means of the entropy current.
    
    This picture was confirmed and derived more rigorously using covariant phase space formalism in a recent paper by Hollands, Kovacs and Reall \cite{Hollands:2022fkn}.  They also showed directly the relationship between $S_\text{inc}$ and $S_\text{Iyer-Wald}$ (reviewed in Appendix \ref{app:zbt}), and they showed the \emph{gauge invariance} to first order of $S_{\text{inc}}$ by changing to different Gaussian null coordinates (gauge invariance was also discussed in \cite{Bhattacharyya:2022njk}). (Also, a generalisation of second law to \emph{second order} was made in the regime of \emph{effective field theories} in their paper, assuming the higher derivative terms are suppressed by powers of some length scale $\ell$ which is small compared to the dynamically relevant length scales, but we will not follow up on that development here.)

    \subsection{Our Higher Curvature / Vector Lagrangian}
    In this article, we consider higher curvature gravity theories non-minimally coupled to vector fields, especially non-gauge vector fields. For example, it might be a non-minimally coupled Proca field.  A general local Lagrangian will take the form:
    \begin{equation}
        L = L(g ^{ab}, R _{abcd}, \nabla_{e_1} R_{abcd}, \cdots, \nabla_{(e_1}\cdots \nabla_{e_p)} R_{abcd}, V_a, \nabla_{b_1} V_a,\cdots, \nabla _{(b_1} \cdots \nabla _{b_q)} V_a)
    \end{equation}
    where it is constructed by all possible contractions of the variables that lead to a covariant scalar. Here, for convenience, we treat $V_a$ as ``fundamental'' rather than $V^a$ to keep track of the contraction using $g ^{ab}$, and we only include the totally symmetrised derivatives on $V_a$ and $R_{abcd}$  since any antisymmetrised second derivative is equivalent to a factor of $R _{abcd}$ by the Ricci identity.  (As we will be working at linear order, we consider this Lagrangian to be valid at arbitrary scales, not an effective field theory.)

    We do not consider vector models which spontaneously break local Lorentz invariance, as for example Einstein-Aether theory \cite{Eling:2004dk}, in which a Lagrange multiplier sets $V^a V_a = -1$.  In such theories, black holes do not have regular bifurcate horizons \cite{Eling:2006ec}, and relatedly the Generalised Second Law does not hold at either the quantum or classical orders \cite{Dubovsky:2006vk,Eling:2007qd}.  As our current work relies essentially on the regularity of the bifurcation surface, such aether models are excluded from the start.

    \subsection{Covariant Phase Space Description of Increasing Entropy}
    In the covariant phase space language, the null-null component of equation of motion $E_{vv}$ is associated with the constraint form $\bm C_k$ with respect to the null translation $k = \partial_v$. We can rephrase the integral of $\delta E_{vv}$ over a time-slice $\mathcal C(v)$ of the horizon by the integration of $\iota_k \bm C_k$ (which is the contraction of $k$ and $\bm C_k$) on $\mathcal C(v)$. On the other hand, the off-shell Noether theorem states that $\bm J_k + \bm C_k = \dd{\bm Q_k}$, where $\bm J_k$ and $\bm Q_k$ are the (off-shell) Noether current and Noether charge for $k$, respectively. Through this equation, the increasing entropy is related (c.f.~section \ref{sec:inc-entp}) to $\bm Q_k$ and the pre-symplectic potential $\bm \Theta$ by 
    \begin{equation}
        \pounds_k \pounds_k \delta S_\text{inc} = - 2\pi \int_{\mathcal C(v)} \dd[D-2]{x}\,\sqrt{h}\,\delta E_{vv} = 2\pi \int_{\mathcal C(v)} \iota_k \delta \bm C_k  = 2\pi \int_{\mathcal C(v)} \delta (\pounds_k \bm Q_k - \iota_k \bm \Theta_k)
    \end{equation}
    where we've used the definition of Noether current $\bm J_k = \bm \Theta_k - \iota_k \bm L$, and we've covariantised the partial derivatives $\partial_v$ to Lie derivatives $\pounds_k$. Under JKM ambiguities, the integrand $\pounds_k \bm Q_k - \iota_k \bm \Theta_k$ is invariant up to an exact differential that integrates to zero for compact $\mathcal C(v)$, as will be shown in Section \ref{sec:jkm}. Hence, our definition of increasing entropy is free of JKM ambiguities.

    \subsection{Inauspicious Terms for the Vector Field}

    Now in the presence of a vector field, we are still tempted to use the method proposed in \cite{Wall:2015raa} to calculate the black hole entropy because such entropy readily satisfies a linearised second law under perturbations of a stationary background. Then it is natural to ask: \emph{does the entropy definition in \cite{Wall:2015raa} safely generalise to the case of non-minimally coupled vector fields?} Recall that in Gaussian null coordinates with $v$ the outgoing null direction, the entropy $S_\text{inc}$ is defined through
    \begin{equation}
        \partial_v^2 S_\text{inc} = - 2\pi \int \dd[D-2]{x} \sqrt{h}\, E_{vv}
    \end{equation}
    on the horizon, where $E_{ab}$ is the equation of motion for $g_{ab}$.  One potential problem is that of the two null indices $v$ of $E_{vv}$, one or both may come from the $V_v$ vector component field rather than the partial derivative $\partial_v$.\footnote{This problem does not arise for a scalar field because it has no indices, and it does not arise for the graviton because of our gauge choice.}  Although the combination $V_v V_v$ does not contribute at the first order of perturbation, terms involving $\partial_v V_v$ might survive even at first order. These would make it seemingly impossible to extract two $\partial_v$'s from $E_{vv}$ altogether. Or, in other words, in the presence of vector fields, we may not be able to obtain a \emph{local} functional expression of $S_\text{inc}$ by integrating $E_{vv}$ with respect to $v$ twice. In mathematical expressions, the off-shell $vv$-component of equation of motion may contain some ``inauspicious terms''\footnote{We thank Diandian Wang for suggesting this name.} in the form of
    \begin{equation}
        \frac{1}{\sqrt{h}}\partial_v \left( \sqrt{h}\, \mathcal{P} \right) \subset \delta E_{vv}\, , \qquad \mathcal{P} = \sum_{I\geq 0} C_I^{i_1 \cdots i_I} D_{(i_1}\cdots D_{i_I)} \delta V_v
    \end{equation}
    where it has only one $v$-derivative available.\footnote{One might try to solve this problem by the Stueckelberg trick, in which the vector field is written as a combination of a gauge field and a scalar: $V_a = A_a - \partial_a \phi$, where the gauge symmetry is $\delta A_a = \partial_a \alpha$, $\delta \phi = \alpha$, so that $V_a$ is gauge-invariant.  Then we could go to a gauge where $A_v = 0$ and pull off an extra $v$ derivative from $\partial_v \phi$.  However, this replaces the previous problem with a new problem of whether the resulting increasing entropy $S_\text{inc}$ is gauge-invariant.  So in this article we will not proceed along these lines.} 

    We will show that, in general, although inauspicious terms involving $\partial_v V_v$ do exist in $E_{vv}$, they do not invalidate the second law, because they sum to an entropy current, which integrates out on a compact horizon time-slice, i.e.,
    \begin{equation}
        \mathcal{P} = D_i J_V^i.
    \end{equation} 

    This will be done as follows: the covariant phase space method shows that the integral of $v\,\delta E_{vv}$ over the \emph{whole} future horizon $\mathcal{H}$, $v\in (-\infty, \infty)$, is zero by assuming compactly supported perturbations:
    \begin{equation}
        \int_{-\infty}^{\infty} \dd{v} \int_{\mathcal{C}(v)} \dd[D-2]{x} \sqrt{h}\, v\, \delta E_{vv} = \int_{\mathcal{H}} \dd{(\delta \bm Q_\xi - \iota_\xi \bm \Theta[\delta]}) = 0
    \end{equation}
    where $\bm Q_\xi$ is the Noether charge of the Killing vector $\xi$ as before, and $\iota_\xi \bm \Theta[\delta]$ is the contraction of $\xi$ with the pre-symplectic potential. Using the structure of $\delta E_{vv}$, namely,
    \begin{equation}
        \delta E_{vv} = \partial_v( \partial_v \varsigma + D_i J^i + \mathcal{P})
    \end{equation}
    (where we've ignored factors of $\sqrt{h}$'s, etc.)~we will show that the integral of the inauspicious terms on the horizon vanishes:
    \begin{equation}
        \int_{-\infty}^{\infty} \dd{v} \int_{\mathcal{C}(v)} \dd[D-2]{x} \sqrt{h}\, \mathcal{P} = 0.
    \end{equation}
    A sufficient --- and as we shall show necessary --- condition for this integral to vanish is if the inauspicious term is a total derivative: $\mathcal{P} = D_i J^i_V$.  But this ensures that we can always pull out a second derivative in order to obtain the entropy $S_\text{inc}$.
    
    One might worry that the assumption of compactly supported perturbations in this calculation is too strong, in the more general case of a noncompact horizon.  But it can be shown that the structure of the coefficients $C^{i_1 \cdots i_I}_I$ in $\mathcal{P}$ depend only on the choice of theory\footnote{At linear order, these coefficients consist of background quantities, i.e., they are unperturbed.}, not the specific perturbation; hence, this total divergence structure continues to hold for an arbitrary perturbation which is not compactly supported in the $\dd v$ direction.\footnote{In the arguments above we also assume that the horizon is spatially compact.  If it is not, and if the perturbation is not, then there may be a flux of entropy coming in from the spatial infinite boundary of the horizon $\partial \cal H$.  But in such situations it is not physically clear that there should be an increasing entropy anyway, since the horizon is no longer a closed system.}

    \subsection{Explicit Entropy Formula}
    
    We also give an explicit calculation for one-derivative vector field non-minimally coupled to $f(\text{Riemann})$ as an example to demonstrate our point. The formula we have obtained for this specific case in Gaussian null coordinates is 
    \begin{equation}\label{vectorS}
        \begin{aligned}
            S_{\text{inc}} & = - 2 \pi \int_{\mathcal{C}(v)} \dd[D-2]{x} \sqrt{h}\, \biggl(4 \pdv{L}{R_{uvuv}} + 16 \pdv{L}{R_{uiuj}}{R_{vkvl}}\bar K_{ij} K_{kl}\\
            & \qquad - 4 \pdv{L}{R_{uiuj}}{(\nabla_v V_v)} \bar K_{ij} V_v - 4 \pdv{L}{(\nabla_u V_u)}{R_{vivj}} V_u K_{ij} + \pdv{L}{(\nabla_u V_u)}{(\nabla_v V_v)} V_u V_v \biggr).
        \end{aligned}
    \end{equation}
    When interpreting this formula, note that the partial derivatives with respect to components of $R_{abcd}$ should be unpacked as
    \begin{equation}
        4 \pdv{R_{uvuv}} \equiv \pdv{R_{uvuv}} - \pdv{R_{uvvu}} + \pdv{R_{vuvu}} - \pdv{R_{vuuv}}
    \end{equation}
    and so on.
    
    The entropy \eqref{vectorS} reduces to the formula for $f(\text{Riemann})$ (equation (\ref{eq:wall-entropy})) when the vector field vanishes, hence, our calculation with vanishing vector fields can be regarded as a clarification of the Appendix A of \cite{Wall:2015raa}. The notations are explained in Section \ref{sec:gnc}. We also explicitly calculated the zero-boost corrections to the Wald formula, and we confirm that they contribute to the entropy current, just as schematically derived in \cite{Bhattacharyya:2021jhr} and \cite{Hollands:2022fkn}.

    \subsection{Example: Generalised Proca field}

    In this article, we are interested in the general picture, where the most general couplings between vector fields and gravity are considered.  It may be that many theories in this class are unphysical for various reasons --- for example, due to negative norm states or instabilities --- but the physical ones are presumably included in this class.  But in order to reassure ourselves that there are at least some potentially physical vector theories to which our results would apply, we can consider a class of vector models which have second order equations of motion and hence are more likely to be physically well-behaved.

    For example, we can consider a theory for a vector field $A_a$ with action
    \begin{equation}
        I = \int \dd[D]{x} \sqrt{-g}\, L
    \end{equation}
    where
    \begin{equation}
        L = \frac{R}{16 \pi G} - \frac{1}{4} F_{ab} F^{ab} - \frac{1}{2} m^2 A_a A^a + \lambda_1 L_\text{GB} + \lambda_2 L_2 + \lambda_3 L_3, \label{eq:nmproca}
    \end{equation}
    with
    \begin{equation}
        \begin{split}
            F_{ab} & = \nabla_a A_b - \nabla_b A_a\\
            L_\text{GB} & = \epsilon^{abef} \epsilon^{cdgh} R_{abcd} R_{efgh}\\
            L_2 & = \frac{1}{2} R A_a A^a + \nabla_a A_b \nabla^b A^a - (\nabla_a A^a)^2\\
            L_3 & = \left( R_{ab} - \frac{1}{2} g_{ab} R \right) \nabla^a A^b
        \end{split}
    \end{equation}
    and three coupling constants $\lambda_1, \lambda_2, \lambda_3$. Here, $\epsilon^{abcd}$ is the rank-4 Levi-Civita tensor, which can be converted to a generalised Kronecker delta. This describes a physical example of a Proca vector field with mass $m$ non-minimally coupled to higher curvature (Einstein-Gauss-Bonnet) gravity. The $A_a$ kinetic term is also corrected at order $\lambda_2$. Note that the $L_{\text{GB}}$ contribution becomes a topological term at $D=4$, but we can also consider $D > 4$.    
    
    The choice above will admit second-order equations by the properties of Lovelock gravity \cite{Lovelock:1971yv} and the properties of generalised Proca theory \cite{Tasinato:2014eka,Heisenberg:2014rta,Heisenberg:2016eld}.  The corrections to Proca in eq.~\eqref{eq:nmproca} are in fact the most general terms possible where the sum of the number of derivatives $\partial_a$ and vectors $A_a$ is 4, and where the equations of motion for $A_a$ and $g_{ab}$ also remain second order in derivatives \cite{Heisenberg:2014rta} --- except that we have omitted the $(A_a A^a)^2$ term because it does not contain any derivatives, and hence will not affect black hole entropy.  ($L_3$ is in fact a total derivative via the Bianchi identity $\nabla^a G_{ab} = 0$, but we keep it since it will give a nontrivial check on our analysis, as the resulting entropy should also be a total derivative on the horizon.  At order 6 there are similar terms multiplying $\nabla^a A^b$ which are not total derivatives.)
    
    Other kinds of non-minimal coupling terms besides those considered above, such as e.g.~
    \begin{equation}
        \lambda_4 R^{abcd} \nabla_a A_c \nabla_b A_d
    \end{equation}
    can result in additional propagating degrees of freedom when $\lambda_4$ is finite, and would thus need to be treated at the level of effective field theory by considering perturbations in $\lambda_4$.

    The inauspicious terms which are potentially problematic arise as the $vv$-component of perturbed equation of motion tensor $\delta E_{vv}$ of this theory contains terms such as
    \begin{equation}
        -\lambda_2 \frac{1}{\sqrt{h}} \partial_v \left( \sqrt{h}\,(A^i\partial_u g_{vi}) \delta A_v \right) \subset \lambda_2 \frac{\delta L_2}{\delta g^{ab}}\subset \delta E_{vv}
    \end{equation}
    when evaluated on the event horizon. However, one does not need to worry about it because all such inauspicious terms add to zero in this specific example, as expected by our general proof.
   
    Using the general entropy formula above, we obtain, in this specific case,
    \begin{equation}
        S_{\text{inc}} = S_{\text{Wald}} + S_{\text{JKM}}
    \end{equation}
    where
    \begin{equation}
        S_{\text{Wald}} = - 2 \pi \int_{\mathcal{C}(v)} \dd[D-2]{x} \sqrt{h} \left( - \frac{1}{8 \pi G} + 8 \lambda_1 \epsilon^{ij}\epsilon^{kl}R_{ijkl} - \lambda_2 A_a A^a + \lambda_3 \nabla_i A^i\right)
    \end{equation}
    and 
    \begin{equation}
        S_{\text{JKM}} = -2 \pi \int_{\mathcal{C}(v)} \dd[D-2]{x}\sqrt{h}\, \left( 32 \lambda_1 \epsilon^{ik} \epsilon^{jl} \bar K_{ij} K_{kl} + 2\lambda_2 A_u A_v - \lambda_3 (K A_u + \bar K A_v) \right)
    \end{equation}
    by identifying $\epsilon^{uvij} \equiv \epsilon^{ij}$, which is the rank-2 Levi-Civita tensor on $\mathcal{C}(v)$.
    
    As said above, $L_3$ is a total derivative term so we should expect that it gives a total derivative contribution to the entropy. This is verified via the identity 
    \begin{equation}
        \nabla_i A^i = D_i A^i + K A_u + \bar K A_v
    \end{equation}
    in Gaussian null coordinates (c.f., Section \ref{sec:gnc}), and the total contribution is just $\lambda_3 D_i A^i$ which can be integrated out on a compact horizon.

    Identifying
    \begin{equation}
        \epsilon^{ij}\epsilon^{kl}(R_{ijkl} + 4 \bar K_{ik} K_{jl}) = 2 R[h] \quad \text{and} \quad - A_a A^a + 2 A_u A_v = - A_i A^i,
    \end{equation}
    on the horizon, where $R[h]$ is the intrinsic Ricci scalar for $h_{ab}$, the total increasing entropy can be written as 
    \begin{equation}
        S_\text{inc} = \frac{\text{Area}[\mathcal{C}(v)]}{4 G} - 2 \pi \int_{\mathcal{C}(v)} \dd[D-2]{x} \sqrt{h}\, \left( 16 \lambda_1 R[h] - \lambda_2 A_i A^i\right).
    \end{equation}
    Hence, the $\lambda_2$ term gives an interesting example of the new vector terms in \eqref{vectorS}, not found by any previous work.  Interestingly, the $\lambda_3$ term does not end up contributing to $S_\text{inc}$ even though it was present in $S_\text{Wald}$.

    \subsection{Plan of Paper}

    The structure of this article will be as follows. In Section \ref{sec:gnc}, we define the Gaussian null coordinates for a stationary black hole background, in order to set up the toolbox for definitions and calculations later.  This section can probably be skimmed by those familiar with the formalism outlined in \cite{Wall:2015raa,Bhattacharyya:2021jhr,Hollands:2022fkn}. In Section \ref{sec:cps}, we introduce the covariant phase space formalism, show how to express the increasing entropy using Noether charge language, and provide a detailed discussion of its invariance under JKM ambiguities. Section \ref{sec:prob-term} demonstrates the validity of the increasing entropy constructed via $\int \delta E_{vv}$ in the case of non-minimally coupled vector fields, using the off-shell structure of the equation of motion and covariant phase space identities. In Section \ref{sec:example1}, we explicitly calculate this entropy for any one-derivative vector field non-minimally coupled to $f(\text{Riemann})$ gravity. The discussion and outlook for future work are given in Section \ref{sec:discussion}.

    \section{Gaussian Null Coordinates} \label{sec:gnc}
    To prepare for our explicit calculation later, we introduce the Gaussian null coordinates. This chart was extensively used in previous analyses \cite{Wall:2015raa,Bhattacharyya:2021jhr,Hollands:2022fkn}.

    \subsection{Gaussian Null Coordinates}

    For a $D$-dimensional spacetime with bifurcate horizon, we can pick one of the horizons (in our case, the future horizon) and construct Gaussian null coordinates (GNC) labelled by $(u,v,x^i)$ with $i = 1,\dots,D-2$, where $u = 0$ marks the horizon picked. In this gauge, the metric can be written as 
	\begin{equation}
        \dd{s^2} = 2 \dd{u} \dd{v} + u^2 F(u,v,x) \dd{v^2} + 2 u \omega_i(u,v,x) \dd{v}\dd{x^i} + h _{ij} (u,v,x) \dd{x^i} \dd{x^j}.
    \end{equation}
    We illustrate the coordinate system in Figure \ref{fig:gnc}.
    
    \begin{figure}[h]
        \centering
        \tikzset{every picture/.style={line width=0.75pt}} 

        \begin{tikzpicture}[x=0.75pt,y=0.75pt,yscale=-1,xscale=1]

        \draw [color={rgb, 255:red, 208; green, 2; blue, 27 }  ,draw opacity=1 ]   (220,50) .. controls (168.85,115.77) and (169.82,134.23) .. (220,200) ;
        \draw [shift={(182.28,119.31)}, rotate = 92.42] [color={rgb, 255:red, 208; green, 2; blue, 27 }  ,draw opacity=1 ][line width=0.75]    (8.74,-2.63) .. controls (5.56,-1.12) and (2.65,-0.24) .. (0,0) .. controls (2.65,0.24) and (5.56,1.12) .. (8.74,2.63)   ;
        \draw [color={rgb, 255:red, 208; green, 2; blue, 27 }  ,draw opacity=1 ]   (91.41,168.59) -- (178.59,81.41) ;
        \draw [shift={(180,80)}, rotate = 135] [color={rgb, 255:red, 208; green, 2; blue, 27 }  ,draw opacity=1 ][line width=0.75]    (10.93,-3.29) .. controls (6.95,-1.4) and (3.31,-0.3) .. (0,0) .. controls (3.31,0.3) and (6.95,1.4) .. (10.93,3.29)   ;
        \draw [shift={(90,170)}, rotate = 315] [color={rgb, 255:red, 208; green, 2; blue, 27 }  ,draw opacity=1 ][line width=0.75]    (10.93,-3.29) .. controls (6.95,-1.4) and (3.31,-0.3) .. (0,0) .. controls (3.31,0.3) and (6.95,1.4) .. (10.93,3.29)   ;
        \draw [color={rgb, 255:red, 208; green, 2; blue, 27 }  ,draw opacity=1 ]   (101.41,91.41) -- (168.59,158.59) ;
        \draw [shift={(168.59,158.59)}, rotate = 45] [color={rgb, 255:red, 208; green, 2; blue, 27 }  ,draw opacity=1 ][line width=0.75]    (10.93,-3.29) .. controls (6.95,-1.4) and (3.31,-0.3) .. (0,0) .. controls (3.31,0.3) and (6.95,1.4) .. (10.93,3.29)   ;
        \draw [shift={(101.41,91.41)}, rotate = 225] [color={rgb, 255:red, 208; green, 2; blue, 27 }  ,draw opacity=1 ][line width=0.75]    (10.93,-3.29) .. controls (6.95,-1.4) and (3.31,-0.3) .. (0,0) .. controls (3.31,0.3) and (6.95,1.4) .. (10.93,3.29)   ;
        \draw [line width=1.5]    (40,220) -- (227.88,32.12) ;
        \draw [shift={(230,30)}, rotate = 135] [color={rgb, 255:red, 0; green, 0; blue, 0 }  ][line width=1.5]    (11.37,-3.42) .. controls (7.23,-1.45) and (3.44,-0.31) .. (0,0) .. controls (3.44,0.31) and (7.23,1.45) .. (11.37,3.42)   ;
        \draw [line width=1.5]    (40,30) -- (227.88,217.88) ;
        \draw [shift={(230,220)}, rotate = 225] [color={rgb, 255:red, 0; green, 0; blue, 0 }  ][line width=1.5]    (11.37,-3.42) .. controls (7.23,-1.45) and (3.44,-0.31) .. (0,0) .. controls (3.44,0.31) and (7.23,1.45) .. (11.37,3.42)   ;
        \draw [color={rgb, 255:red, 74; green, 144; blue, 226 }  ,draw opacity=1 ]   (135,125) ;
        \draw [shift={(135,125)}, rotate = 0] [color={rgb, 255:red, 74; green, 144; blue, 226 }  ,draw opacity=1 ][fill={rgb, 255:red, 74; green, 144; blue, 226 }  ,fill opacity=1 ][line width=0.75]      (0, 0) circle [x radius= 3.35, y radius= 3.35]   ;
        \draw [color={rgb, 255:red, 208; green, 2; blue, 27 }  ,draw opacity=1 ]   (50,50) .. controls (101.52,116.13) and (101.52,134.6) .. (50,200) ;
        \draw [shift={(88.39,130.33)}, rotate = 272.1] [color={rgb, 255:red, 208; green, 2; blue, 27 }  ,draw opacity=1 ][line width=0.75]    (8.74,-2.63) .. controls (5.56,-1.12) and (2.65,-0.24) .. (0,0) .. controls (2.65,0.24) and (5.56,1.12) .. (8.74,2.63)   ;
        \draw [color={rgb, 255:red, 208; green, 2; blue, 27 }  ,draw opacity=1 ]   (60,210) .. controls (126.13,159.21) and (144.6,159.58) .. (210,210) ;
        \draw [shift={(129.58,172.34)}, rotate = 357.03] [color={rgb, 255:red, 208; green, 2; blue, 27 }  ,draw opacity=1 ][line width=0.75]    (8.74,-2.63) .. controls (5.56,-1.12) and (2.65,-0.24) .. (0,0) .. controls (2.65,0.24) and (5.56,1.12) .. (8.74,2.63)   ;
        \draw [color={rgb, 255:red, 208; green, 2; blue, 27 }  ,draw opacity=1 ]   (60,40) .. controls (126.13,91.52) and (144.6,91.52) .. (210,40) ;
        \draw [shift={(140.33,78.39)}, rotate = 177.9] [color={rgb, 255:red, 208; green, 2; blue, 27 }  ,draw opacity=1 ][line width=0.75]    (8.74,-2.63) .. controls (5.56,-1.12) and (2.65,-0.24) .. (0,0) .. controls (2.65,0.24) and (5.56,1.12) .. (8.74,2.63)   ;
        \draw    (60,120) ;
        \draw [shift={(60,120)}, rotate = 0] [color={rgb, 255:red, 0; green, 0; blue, 0 }  ][fill={rgb, 255:red, 0; green, 0; blue, 0 }  ][line width=0.75]      (0, 0) circle [x radius= 3.35, y radius= 3.35]   ;
        \draw [color={rgb, 255:red, 139; green, 87; blue, 42 }  ,draw opacity=1 ]   (153.5,107) ;
        \draw [shift={(153.5,107)}, rotate = 0] [color={rgb, 255:red, 139; green, 87; blue, 42 }  ,draw opacity=1 ][fill={rgb, 255:red, 139; green, 87; blue, 42 }  ,fill opacity=1 ][line width=0.75]      (0, 0) circle [x radius= 3.35, y radius= 3.35]   ;

        \draw (240,25) node    {$v$};
        \draw (241,217.5) node    {$u$};
        \draw (153,87) node    {$\mathcal{H}$};
        \draw (134.5,140.5) node  [color={rgb, 255:red, 74; green, 144; blue, 226 }  ,opacity=1 ]  {$\mathcal{B}$};
        \draw (240,123.5) node  [color={rgb, 255:red, 208; green, 2; blue, 27 }  ,opacity=1 ]  {$\xi =v\partial _{v} -u\partial _{u}$};
        \draw (35.5,120.5) node    {$\left\{x^{i}\right\}$};
        \draw (162.5,123) node  [color={rgb, 255:red, 139; green, 87; blue, 42 }  ,opacity=1 ]  {$\mathcal{C}( v)$};

        \end{tikzpicture}
        \caption{Gaussian null coordinates. Here, $\mathcal{H}$ is the black hole horizon ($u=0$), $\mathcal{B}$ is the bifurcation surface of the stationary background ($u=v=0$), $\mathcal{C}(v)$ is a time-slice of $\mathcal{H}$ at null time $v$. } \label{fig:gnc}
    \end{figure}
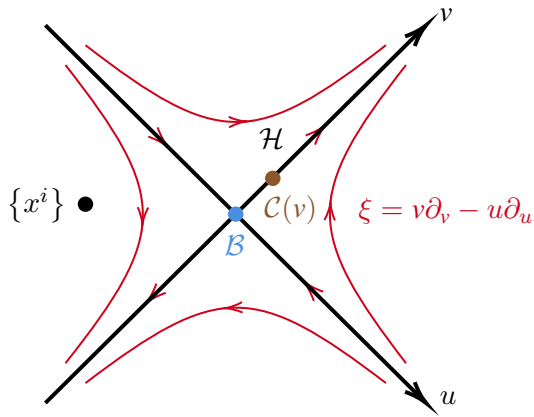
    
    For later convenience, we also calculate the inverse of such metric:
	\begin{equation}
	g ^{ab}=\mqty(u^2(-F + \omega_i \omega^i) & 1 & -u \omega^i\\1 & 0 & 0\\-u \omega^i & 0 & h ^{ij}).
	\end{equation}
    where $h^{ij}$ is the inverse of $h_{ij}$ and $\omega^i = h^{ij} \omega_j$. 

    Note, in this article, we only consider a \emph{fixed} GNC, and only consider the quantities evaluated \emph{on} the horizon. For a more detailed analysis involving coordinate transformations/off-horizon behaviours, the readers can refer to \cite{Bhattacharyya:2022njk,Hollands:2022fkn}.
	
	We consider a $D$-dimensional stationary black hole spacetime equipped with a Killing vector 
    \begin{equation}
        \xi = v \partial_v - u \partial_u
    \end{equation}
    as our background. Then in our calculation, we add dynamical (non-stationary) linear perturbations to our metric while \emph{preserving the gauge choice}. More explicitly, the metric is
    \begin{equation}
        g_{ab} = \bar g_{ab} + \epsilon \delta g_{ab}
    \end{equation}
	where $\bar g_{ab}$ is the \emph{stationary} metric which already satisfies the gauge choice and the stationarity condition $\pounds_\xi \bar g_{ab} = 0$, $\delta g_{ab}$ is the \emph{non-stationary} perturbation (i.e., $\pounds_\xi \delta g_{ab} \neq 0$) preserving $\delta g_{uv} = \delta g_{uu} = \delta g_{ui} = 0$, $\delta g_{vv} = u^2\delta F$, $\delta g_{vi} = u\delta \omega_i$, and $\epsilon \ll 1$ is a small perturbation parameter.
    
    According to the form of the Killing vector in our chosen coordinate system, we can interpret the stationarity condition as the boost invariance under the rescaling $(u,v) \mapsto (au, v/a)$, so the functional dependence on $(u,v)$ of the background and perturbations is 
    \begin{equation}
        \bar g_{ab} (u,v,x^i) = \bar g_{ab} (uv,x^i) \quad \text{but} \quad \delta g_{ab} (u,v,x^i) \neq \delta g_{ab} (uv,x^i),
    \end{equation}
    by arguments in \cite{Bhattacharyya:2021jhr,Hollands:2022fkn}. 
    
    In the following text, we call a tensor field $T$ \emph{stationary} if it satisfies $\pounds_\xi T = 0$. We sometimes use a bar ``\,$\bar{~}$\,'' to denote stationary background quantities.
	
    \subsection{Killing Weight and Perturbation Accounting}
    The stationary GNC is invariant under rescaling $u \mapsto au, v \mapsto v/a$ with constant $a \in \mathbb{R}_+$ (positivity to preserve the time direction). Such rescaling leads to certain scale transformations for components of any tensor field, when resolved in GNC. Consider a general tensor field $T_{a_1 \cdots a_n}$ (constructed out of dynamical fields of the theory), we extract its components by defining the following notation 
    \begin{equation}
        T_{\left\{ \mathfrak{e} \right\}} := T_{a_1 \cdots a_n} \mathfrak{e}_1^{a_1} \cdots \mathfrak{e}_n^{a_n}, \qquad \mathfrak{e}_1, \cdots, \mathfrak{e}_n \in \left\{ \pdv{u}, \pdv{v}, \pdv{x^i} \right\}.
    \end{equation}
    Then, the \emph{Killing weight} or \emph{boost weight} $w$ of some tensor component $T_{\left\{ \mathfrak{e} \right\}}$ is defined by how it transforms under such rescaling: 
    \begin{equation}
        T_{\left\{ \mathfrak{e} \right\}} \mapsto a^{w} T_{\left\{ \mathfrak{e} \right\}}, \quad u \mapsto au, v\mapsto v/a.
    \end{equation}
    In other words, the Killing weight $w$ of $T_{\left\{ \mathfrak{e} \right\}}$ can be expressed as 
    \begin{equation}
        w = \text{number of $\pdv{v}$ in $\left\{ \mathfrak{e} \right\}$} - \text{number of $\pdv{u}$ in $\left\{ \mathfrak{e} \right\}$},
    \end{equation}
    i.e., the number of downstairs $v$-indices minus the number of downstairs $u$-indices on $T_{\left\{ \mathfrak{e} \right\}}$.

    One may wonder why we are interested in such Killing weight. As we will show below, the Killing weight gives a neat accounting for linear perturbations around stationary background. 

    Before we proceed, we state a lemma (proved in \cite{Bhattacharyya:2021jhr}) which will be used later: 
    \begin{lem}\label{lem:Lie}
        In GNC, we can write the component of the $\xi$-Lie derivative of a weight $w$ tensor $T$ with respect to $\{\mathfrak{e}\}$ as 
        \begin{equation}
            \left( \pounds_\xi T \right)_{\left\{ \mathfrak{e} \right\}} = \left( v \partial_v - u \partial_u + w \right) T_{\left\{ \mathfrak{e} \right\}}.
        \end{equation}
    \end{lem}
    Let's investigate what a Killing symmetry will imply for the values of components of a tensor field $T_{a_1 \cdots a_n}$ which satisfies stationarity condition $\pounds_\xi T = 0$.  We fulfil the discussion by proving the following claim.

    \begin{clm}
        In GNC, assuming regularity at the bifurcation surface, any positive Killing weight component of a stationary tensor vanishes on a Killing horizon. \label{clm:2}
    \end{clm}

    \begin{proof}
        Consider a general tensor field $T_{a_1 \cdots a_n}$ which satisfies the stationary condition $\pounds_\xi T = 0$. Take a weight $w$ component $T_{\left\{ \mathfrak{e} \right\}}$ and use Lemma \ref{lem:Lie} to get 
        \begin{equation}
            \left( v \partial_v - u \partial_u + w \right) T_{\left\{ \mathfrak{e} \right\}} = 0.
        \end{equation}
        Assuming everything behaves smoothly near the horizon, we can find the solution of above equation as 
        \begin{equation}
            T_{\left\{ \mathfrak{e} \right\}}(u,v,x^i) = \sum_{k \geq -w} u^{k+w} v^k C_k(x^i)
        \end{equation}
        where $C_k(x^i)$ are functions of the codimension-2 $\{x^i\}$ coordinates only. Evaluating on the horizon $u=0$, we have 
        \begin{equation}
            T_{\left\{ \mathfrak{e} \right\}}(0,v,x^i) = v^{-w} C_{-w}(x^i),
        \end{equation}
        which is divergent at $v=0$ for positive weight $w>0$ unless $C_{-w}(x^i) = 0$. Hence,
        \begin{equation}
            T_{\left\{ \mathfrak{e} \right\}}(0,v,x^i) = 0 \quad \text{for} \quad w>0
        \end{equation}
        where $C_{-w}(x^i)$ must be zero to be compatible with the regularity assumption.
    \end{proof}
    
    Now we consider the scenario where linear perturbations are turned on. Any tensor field $T$ in consideration is now 
    \begin{equation}
        T = \bar T + \epsilon\, \delta T
    \end{equation}
    where $\bar T$ is stationary but $\delta T$ is not. Then, the above claim gives a neat accounting of linear perturbations through the following corollary: 
    \begin{cor}
        In GNC, assuming regularity, any positive Killing weight component of a linearly perturbed stationary tensor is of at least linear order in the perturbation parameter on the horizon $u=0$. We can also express this as 
        \begin{equation}
            T_{\left\{ \mathfrak{e} \right\}}(0,v,x^i) \sim \mathcal{O}(\epsilon) \quad \text{for} \quad w>0.
        \end{equation}
    \end{cor}
    This is convenient for our later analysis at linear order in $\epsilon$ as anything that is manifestly a product of positive-weight components is of $\mathcal{O}(\epsilon^2)$ at least. We can ignore such contribution in our calculations in the scope of this text. As one can see, this keeps track of first-order perturbation nicely and it very much simplifies relevant calculations.

    Before we move on, there is one caveat with this corollary that we have to mention: the metric perturbations actually has $\delta g_{vv} = \delta g_{vi} = 0$ on the horizon, this is because the GNC impose further constraints on $\delta g_{ab}$.

    \subsection{GNC Decomposition of Tensor Components}

    Here, we state how we can decompose any component of a tensor $T_{a_1 \cdots a_n}$ that is constructed out of the field contents of the theory (i.e., the inverse metric $g^{ab}$, the Riemann tensor $R_{abcd}$ and the matter field $\phi_{\mathfrak{I}}$ where $\mathfrak{I}$ label the tensor indices) under GNC. This is well-investigated in previous publications \cite{Bhattacharyya:2021jhr,Hollands:2022fkn} so we just quote the results.

    \begin{clm}\label{clm:4}
        Consider a tensor $T_{a_1 \cdots a_n}$ constructed out of $g^{ab}, R_{abcd}$ and $\phi_{\mathfrak{I}}$ and their covariant derivatives. Then, under GNC, any component $T_{\left\{ \mathfrak{e} \right\}}$ of it with respect to $\left\{ \mathfrak{e} \right\}$ can be written as an equal-Killing-weight polynomial constructed from $h^{ij}$-contractions of factors 
        \begin{equation}
            \partial_u^{p} \partial_{v}^q D_{(i_1} \cdots D_{i_r)} \psi, \quad u^N
        \end{equation}
    for $p,q,r,N\in \mathbb{N}_0$, where $\psi\in \{F, \omega_i, h_{ij}, R_{ijkl}[h], \phi_{\mathfrak{I}}\}$.
    \end{clm}

    This decomposition is very useful when one considers the positive-Killing-weight components on the horizon at linear order of perturbation, and it further simplifies as all terms proportional to $u^N$, $N>0$ vanish. In this case, the GNC and Killing weight are used to account for the structure of the perturbed quantities in the following way:
    \begin{cor}\label{cor:5}
        Consider a tensor $T_{a_1\cdots a_n}$ constructed out of $g^{ab}, R_{abcd}, \phi_{\mathfrak{I}}$ which are stationary. Then, at first order, any component of perturbed $\delta T_{\left\{ \mathfrak{e} \right\}}$ with respect to $\left\{ \mathfrak{e} \right\}$ with Killing weight $w > 0$ can be decomposed on the horizon as 
        \begin{equation}
            \delta T_{\left\{ \mathfrak{e} \right\}} = \sum_I \sum_{k \geq 0} \sum_{m = 0}^{w+k} A^I_{(-k)}  \partial_v^{k+w-m} \delta B^I_{(m)} 
        \end{equation}
        where $A^I_{(-k)}$ are weight-$(-k)$ polynomials constructed out of $\partial_{u}^{p} \partial_{v}^{q} D_{(i_1} \cdots D_{i_r)} \psi_{(p-q-k)}$ and $\delta B^I_{(m)}$ are weight-$m$  perturbed quantities characterised by 
        \begin{equation}
            \delta B^I_{(0)} \in \{\partial_u^p \partial_v^q D_{(i_1}\cdots D_{i_r)} \delta \psi_{(p-q)}\}
        \end{equation}
        and
        \begin{equation}
            \delta B^I_{(m>0)} \in \{\partial_u^p D_{(i_1}\cdots D_{i_r)} \delta \psi_{(m+p)}\}
        \end{equation}
        with $\psi_{(p)}$ a weight-$p$ component in $\left\{ F, \omega_i, h_{ij}, R_{ijkl}[h], \phi_{\mathfrak{I}} \right\}$. Here, $p,q,r \in \mathbb{N}_0$, and $I$ labels different terms with the same properties. This follows from Claims \ref{clm:2} and \ref{clm:4}.
    \end{cor}
    
    Note that any positive weight $p$ of $\psi_{(p)}$ comes from components of $\phi_{\mathfrak{I}}$ with more $v$ indices than $u$ indices. This suggests that any ``net'' $v$ index on some tensor component comes from either a $\partial_v$ or an extra $v$ index on the matter field $\phi_{\mathfrak{I}}$. In the case of a vector field $V_a$, any ``net'' $v$ comes from either a $\partial_v$ or a $V_v$. We will then use the above results to analyse the structure of the variation of the equation of motion tensor $\delta E_{vv}$ on the horizon. 
    
    A more detailed analysis on GNC decomposition can be found in \cite{Hollands:2022fkn}.
    
	\subsection{Useful Geometrical Quantities in GNC}

    For later convenience, we calculate the connection coefficients and certain components of the covariant derivative in advance.

	\subsubsection*{Connection Coefficients on the Horizon}

	The non-vanishing connection coefficients on the horizon are
	\begin{equation}
	\text{Killing weight 0:}\quad\conn{k}{i}{j}, \conn{v}{v}{i}, \conn{u}{u}{i}, \conn{i}{u}{v}, \qquad \text{K.w.\ $-1$:}\quad \conn{v}{i}{j}, \conn{j}{i}{u}, \qquad \text{K.w.\ 1:} \quad \conn{u}{i}{j}, \conn{i}{j}{v}.
	\end{equation}
	
	We calculate on the horizon:
    \bgroup \allowdisplaybreaks[0]
    \begin{subequations}
	\begin{alignat}{3}
		&\conn{u}{i}{j} = - K _{ij} \sim \mathcal{O}(\epsilon) \qquad
		&&\conn{j}{i}{v} = + K\indices{_i^j} \sim \mathcal{O}(\epsilon) \qquad
		&&\conn{v}{v}{i} = - \frac{1}{2} \omega_i \sim \mathcal{O}(1)\\
		&\conn{v}{i}{j} = - \bar K _{ij} \sim \mathcal{O}(1)
		&&\conn{j}{i}{u} = + \bar K\indices{_i^j} \sim \mathcal{O}(1)
		&&\conn{u}{u}{i} = + \frac{1}{2} \omega_i \sim \mathcal{O}(1)\\
		&\conn{i}{u}{v} = + \frac{1}{2} \omega^i \sim \mathcal{O}(1)
		&&
		&&\conn{i}{j}{k} = \text{Christoffel Symbols for }h _{ij} \sim \mathcal{O}(1)
	\end{alignat}
    \end{subequations}
    \egroup
    and 
	\begin{equation}
		\conn{a}{u}{u} = \conn{a}{v}{v} = \conn{u}{v}{a} = \conn{v}{u}{a} = 0
    \end{equation}
	where $K_{ij} = \frac{1}{2} \partial_v h _{ij}$ and $\bar K_{ij} = \frac{1}{2} \partial_u h _{ij}$ are the extrinsic curvature with respect to $v$ and $u$. We also define the trace: $K = h^{ij} K_{ij}$ and $\bar K = h^{ij} \bar K_{ij}$.

	\subsubsection*{Components of Covariant Derivative on the Horizon}
    \label{sec:table}
    On the horizon, the components of $v$- and $i$-covariant derivatives of a tensor $T$ are expressed in table \ref{tb:covder}.
	\begin{table}[h]
	\begin{tabular}{|c|c|c|}
		\hline
		$\nabla$ & $T^{\cdots a\cdots }$ & $T_{\cdots a \cdots}$ \\
		\hline
		& & \\
		$v$ & $\begin{aligned}
			\nabla_v T^{\cdots v\cdots} & = \partial_v T^{\cdots v\cdots} - \frac{1}{2} \omega_i T^{\cdots i\cdots}\\
			\nabla_v T^{\cdots u\cdots} & = \partial_v T^{\cdots u\cdots} \\
			\nabla_v T^{\cdots i\cdots} & = \partial_v T^{\cdots i\cdots} + \frac{1}{2} \omega^i T^{\cdots u\cdots} + K\indices{^i _{j}}T^{\cdots j\cdots}
		\end{aligned}$ & $\begin{aligned}
			\nabla_v T_{\cdots u\cdots} & = \partial_v T_{\cdots u\cdots} - \frac{1}{2} \omega^i T_{\cdots i\cdots} \\
			\nabla_v T_{\cdots v\cdots} & = \partial_v T_{\cdots v\cdots}\\
			\nabla_v T_{\cdots i\cdots} & = \partial_v T_{\cdots i\cdots} + \frac{1}{2} \omega_i T_{\cdots v\cdots} - K\indices{_i^j}T_{\cdots j\cdots}
		\end{aligned}$\\
		& & \\
		\hline
		& & \\
		$i$ & $\begin{aligned}
			\nabla_i T^{\cdots v\cdots} & = D_i T^{\cdots v\cdots} - \bar K _{ij}T^{\cdots j\cdots} - \frac{1}{2} \omega_i T^{\cdots v\cdots}\\
			\nabla_i T^{\cdots u\cdots} & = D_i T^{\cdots u\cdots} - K _{ij}T^{\cdots j\cdots} + \frac{1}{2} \omega_i T^{\cdots u\cdots}\\
			\nabla_i T^{\cdots j\cdots} & = D_i T^{\cdots j\cdots} + K\indices{_i^j} T^{\cdots v\cdots} + \bar K\indices{_i^j} T^{\cdots u\cdots}
		\end{aligned}$ & $\begin{aligned}
			\nabla_i T_{\cdots u\cdots} & = D_i T_{\cdots u\cdots} - \bar K\indices{_i^j}T_{\cdots j\cdots} - \frac{1}{2} \omega_i T_{\cdots u\cdots}\\
			\nabla_i T_{\cdots v\cdots} & = D_i T_{\cdots v\cdots} - K\indices{_i^j}T_{\cdots j\cdots} + \frac{1}{2} \omega_i T_{\cdots v\cdots}\\
			\nabla_i T_{\cdots j\cdots} & = D_i T_{\cdots j\cdots} + \bar K\indices{_{ij}}T_{\cdots v\cdots} + K\indices{_{ij}} T_{\cdots u\cdots}
		\end{aligned}$ \\
		& & \\
		\hline
	\end{tabular}
    \caption{Components of Tensor Covariant Derivatives in GNC.} \label{tb:covder}
	\end{table}
    Here, $D_i$ is the intrinsic covariant derivative in the codimension-2 directions $x^i$. 

    For the sake of gauge invariance when we consider GNC transformations such as non-constant scaling $v = a(x')v', u = u'/a(x'), x = x'$, the twist $\omega_i$ can be thought as a connection 1-form of the 2-dimensional normal bundle on the codimension-2 space. We can combine it with $D_i$ to define the \emph{gauge covariant derivatives} 
    \begin{equation}
        \mathcal{D}_i = D_i + \frac{1}{2} \omega_i W
    \end{equation}
    where $W$ is the Killing weight operator such that it reads off the weight of any tensor component in GNC, e.g., $W T_{(k)} = k T_{(k)}$. Expressions written in terms of such derivatives should be gauge covariant/invariant. There is a more detailed discussion in \cite{Hollands:2022fkn}.

    
    \section{Covariant Phase Space Description of the Increasing Entropy}\label{sec:cps}

    In this section, we first introduce the covariant phase space formalism as a major toolkit for our investigation; we then reformulate Wall's entropy proposal \cite{Wall:2015raa} in this language, and we discuss the differences and relations between notions of JKM ambiguities appeared in \cite{Jacobson:1993vj} (as differential form ambiguities) and in \cite{Wall:2015raa} (as a collection of negative-positive boost-weight quantities). 
    
    \subsection{Covariant Phase Space Formalism with Vector Fields}
    Start with the Lagrangian form $\bm L = L \bm \epsilon$ where $\bm \epsilon$ is the volume form. The Lagrangian density $L$ is the most general one of interest 
    \begin{equation}
        L = L(g ^{ab}, R _{abcd}, \nabla_{e_1} R_{abcd}, \cdots, \nabla_{(e_1}\cdots \nabla_{e_p)} R_{abcd}, V_a, \nabla_{b_1} V_a,\cdots, \nabla _{(b_1} \cdots \nabla _{b_q)} V_a).
    \end{equation}
    We neglected other types of matter fields for the purpose of this paper, but the covariant phase space formalism for other types of matter can easily be derived. 
    
    The variation of Lagrangian can be written as 
    \begin{equation}
        \delta \bm L = \frac{1}{2} \bm E_{ab} \delta g^{ab} + \bm{\mathcal{E}}^a \delta V_a + \dd{\bm \Theta}[\phi,\delta \phi]
    \end{equation}
    where $\bm E^{ab} = E^{ab} \bm \epsilon$ and $\bm{\mathcal{E}}^a = \mathcal{E}^a \bm \epsilon$ are the equations of motion forms for $g$ and $V$, respectively, $\bm \Theta$ is the pre-symplectic potential, and $\phi\equiv (g,V)$ is a convenient collective notation for the dynamical fields. Note that we have adopted the following convention:
    \begin{equation}
        E_{ab} = \frac{2}{\sqrt{-g}} \frac{\delta I}{\delta g^{ab}}
    \end{equation}
    where $I = \int \bm L$ is the action. 

    In the case where we vary $\bm L$ through a diffeomorphism generated by vector field $\zeta$, we replace $\delta$ with $\pounds_\zeta$ to obtain 
    \begin{equation}
        \pounds_\zeta \bm L = \frac{1}{2} \bm E_{ab} \pounds_\zeta g^{ab} + \bm{\mathcal{E}}^a \pounds_\zeta V_a + \dd{\bm \Theta}_\zeta
    \end{equation}
    by denoting $\bm \Theta_\zeta \equiv \bm \Theta[\phi, \pounds_\zeta \phi]$. We then use the Cartan-Killing equation
    \begin{equation}
        \pounds_\zeta = \iota_\zeta \mathrm{d} + \mathrm{d}\iota_\zeta
    \end{equation}
    where $\iota_\zeta$ is the contraction with $\zeta$ w.r.t.~the first index, and the expressions of Lie derivatives
    \begin{equation}
        \pounds_\zeta g^{ab} = - 2 \nabla^{(a} \zeta^{b)}, \quad \pounds_\zeta V_a = \zeta^b \nabla_b V_a + V_b \nabla_a \zeta^b
    \end{equation}
    to obtain
    \begin{equation}
        \dd{\bm J_\zeta} = \left(E_{ab} \nabla^a \zeta^b - \mathcal{E}^a (\zeta^b \nabla_b V_a + V_b \nabla_a \zeta^b)\right) \bm \epsilon \label{eq:dJ}
    \end{equation}
    by defining the \emph{off-shell Noether current}
    \begin{equation}
        \bm J_\zeta = \bm \Theta_\zeta - \iota_\zeta \bm L.
    \end{equation}

    Taking the Hodge dual of equation \eqref{eq:dJ}, we get 
    \begin{equation}
    \begin{split}
            -\nabla^a (\star J_\zeta)_a & = E_{ab} \nabla^a \zeta^b - \mathcal{E}^a \zeta^b \nabla_b V_a - \mathcal{E}^a V_b \nabla_a \zeta^b\\
            & = \nabla^a (E_{ab} \zeta^b - \mathcal{E}_a V_b \zeta^b) - \zeta^b \nabla^a E_{ab} - \mathcal{E}^a \zeta^b \nabla_b V_a + V_b \zeta^b \nabla_a \mathcal{E}^a + \mathcal{E}^a \zeta^b \nabla_a V_b,
        \end{split}
    \end{equation}
    which can be rearranged to be
    \begin{equation}
        \nabla^a\left((\star J_\zeta)_a + E_{ab} \zeta^b - \mathcal{E}_a V_b \zeta^b\right) = \zeta^b \left(\nabla^a E_{ab} - \mathcal{E}^a F_{ab} - V_b\nabla_a \mathcal{E}^a\right)
    \end{equation}
    where $F_{ab} = \nabla_a V_b - \nabla_b V_a$.

    Integrating over whole spacetime and assuming the variations are compactly supported\footnote{We are free to do so as we are off-shell.}, we have 
    \begin{equation}
        \int \dd[D]{x} \sqrt{-g}\, \zeta^b \left(\nabla^a E_{ab} - \mathcal{E}^a F_{ab} - V_b\nabla_a \mathcal{E}^a\right) = 0,
    \end{equation}
    which holds for arbitrary $\zeta$. This admits a \emph{generalised Bianchi identity}
    \begin{equation}
        \nabla^a E_{ab} = \mathcal{E}^a F_{ab} + V_b\nabla_a \mathcal{E}^a.
    \end{equation}

    Accordingly, we get 
    \begin{equation}
        \nabla^a\left((\star J_\zeta)_a + E_{ab} \zeta^b - \mathcal{E}_a V_b \zeta^b\right) = 0
    \end{equation}
    and we define the \emph{constraint form}, a $(D-1)$-form $\bm C_\zeta$ via its Hodge dual 
    \begin{equation}
        (\star C_\zeta)_a = E_{ab} \zeta^b - \mathcal{E}_a V_b \zeta^b
    \end{equation}
    so 
    \begin{equation}
    \begin{split}
            (C_\zeta)_{a_1 \cdots a_{D-1}} & = (-1)^D (\star \star C_\zeta)_{a_1 \cdots a_{D-1}}\\
            & = (-1)^D \epsilon_{a_1 \cdots a_{D-1} a} \left(E^{ab}\zeta_b - \mathcal{E}^a V^b \zeta_b\right)\\
            & = (-1)^{2D-1} \left(E^{ab}\zeta_b - \mathcal{E}^a V^b \zeta_b\right) \epsilon_{a a_1 \cdots a_{D-1}}\\
            & = \left(- E^{ab}\zeta_b + \mathcal{E}^a V^b \zeta_b\right) \epsilon_{a a_1 \cdots a_{D-1}}.
        \end{split}
    \end{equation}

    Then we have 
    \begin{equation}
        \dd{(\bm J_\zeta + \bm C_\zeta)} = 0,
    \end{equation}
    which, by the Poincar\'e lemma, admits an \emph{off-shell Noether charge} $\bm Q_{\zeta}$, i.e.,
    \begin{equation}
        \bm J_\zeta + \bm C_\zeta = \dd{\bm Q_\zeta}. \label{eq:JplusC}
    \end{equation}

    The variation of this equation will be the key to uncover the covariant phase space description of the increasing entropy, which is shown in the next subsection. 

    Alternatively, we can rewrite the variation of $\bm J_\zeta$ as follows:
    \begin{equation}
    \begin{split}
            \delta \bm J_\zeta & = \delta \bm \Theta_\zeta - \iota_\zeta \delta \bm L\\
            & = \delta \bm \Theta_\zeta - \iota_\zeta \left(\frac{1}{2} \bm E_{ab} \delta g^{ab} + \bm{\mathcal{E}}^a \delta V_a + \dd{\bm \Theta[\phi, \delta \phi]}\right)\\
            & = \delta \bm \Theta_\zeta - \pounds_\zeta \bm \Theta[\phi, \delta \phi] + \dd{\iota_\zeta \bm \Theta[\phi, \delta \phi]} - \iota_\zeta \left(\frac{1}{2} \bm E_{ab} \delta g^{ab} + \bm{\mathcal{E}}^a \delta V_a\right)\\
            & = \bm \omega[\phi; \delta \phi, \pounds_\zeta \phi] + \dd{\iota_\zeta \bm \Theta[\phi, \delta \phi]} - \iota_\zeta \left(\frac{1}{2} \bm E_{ab} \delta g^{ab} + \bm{\mathcal{E}}^a \delta V_a\right)
        \end{split}
    \end{equation}
    where we assumed that the diffeomorphism generating vector is field independent (i.e., $\delta \zeta = 0$) and used the definition of the \emph{off-shell pre-symplectic current}
    \begin{equation}
        \bm \omega[\phi; \delta \phi, \pounds_\zeta \phi] = \delta \bm \Theta_\zeta - \pounds_\zeta \bm \Theta[\phi, \delta \phi].
    \end{equation} 

    Plugging in \eqref{eq:JplusC}, we get an important identity
    \begin{equation}
        \delta \bm C_\zeta = \dd{(\delta \bm Q_\zeta - \iota_\zeta \bm \Theta[\phi, \delta \phi])} - \bm \omega[\phi; \delta \phi, \pounds_\zeta \phi] + \iota_\zeta \left(\frac{1}{2} \bm E_{ab} \delta g^{ab} + \bm{\mathcal{E}}^a \delta V_a\right). \label{eq:magic-id} 
    \end{equation}

    In Section \ref{sec:prob-term}, we will use this equation to resolve the issues raised by the inauspicious vector terms, generalising the procedure studied in \cite{Hollands:2022fkn}. A review of the treatment of zero-boost terms based on this identity in \cite{Hollands:2022fkn} is given in Appendix \ref{app:zbt}.

    \subsection{Increasing Entropy}\label{sec:inc-entp}
    The original definition of the increasing entropy is through the second-derivative expression 
    \begin{equation}
        \partial_v^2 \delta S_{\text{inc}} = - 2 \pi \int_{\mathcal{C}(v)} \dd[D-2]{x}\sqrt{h}\, \delta E_{vv}.
    \end{equation}
    We now show this can be recast covariantly in terms of the Noether charge $\bm Q_k$ and pre-symplectic potential $\bm \Theta_k$ with respect to the null translation $k = \pdv{v}$ by taking the $k$-contraction of the constraint form $\delta \bm C_k$ and apply the variation of equation \eqref{eq:JplusC}:
    \begin{equation} \label{eq:Sinc-cps-1}
        \pounds_k \pounds_k \delta S_\text{inc} = 2 \pi \int_{\mathcal{C}(v)} \iota_k \delta \bm C_k = 2 \pi \int_{\mathcal{C}(v)} \delta \left(\pounds_k \bm Q_k - \iota_k \bm \Theta_k\right)
    \end{equation}
    where we've identified the $v$-derivatives as Lie derivatives w.r.t.~$k$, we've used 
    \begin{equation}
        \iota_k \delta \bm C_k = \iota_k (\dd{\delta \bm Q_k} - \delta \bm J_k) = \pounds_k \delta \bm Q_k - \dd{\iota_k \delta \bm Q_k} - \iota_k (\delta \bm \Theta_k - \iota_k \delta \bm L)
    \end{equation}
    and integrated out $d\iota_k \delta \bm Q_k$ on the horizon time-slice. Note that the double contraction $\iota_k \iota_k \delta \bm L$ vanishes by the antisymmetry of differential forms. A subsequent GNC boost-weight analysis can be carried out for $\delta \bm Q_k$ and $\iota_k \delta \bm \Theta_k$, and it should follow that $\delta \bm Q_k$ contains one extra $\pounds_k$, and $\iota_k \delta \bm \Theta_k$ contains two extra $\pounds_k$'s. These statements should be true modulo exact forms\footnote{This is related to the possibility of an entropy current.}, which integrate out on compact horizons. Then, the entropy can be uniquely extracted by pulling two Lie derivatives out.  In this paper, we will not delve further in this direction as we will directly calculate the entropy formula by extracting two $\partial_v$'s from $\delta E_{vv}$ in GNC as before. 

    We also notice that the above expression is exact in the variation $\delta$ as the null translation vector $k$ is constructed to be field independent. This confirms the expected exactness of the (seemingly non-exact) zero-boost corrections mentioned in the Introduction. Although everything above is exact in $\delta$, we explicitly kept $\delta$'s everywhere to emphasise that the procedure of pulling out two derivatives (covariantly as $\pounds_k$) is valid only at the linear order of perturbation. 

    \subsubsection{Minimal Matter Sources}\label{sec:min-cpl} In the special cases where the matter fields are minimally coupled to gravity, we can split the Lagrangian into two sectors: the gravitational sector $\bm L^\text{g}$ and the matter sector $\bm L^\text{m}$, and balance the first order perturbation of gravity against higher order terms of matter.  This allows us to consistently turn on a nonzero stress tensor $\bm T_{ab}$ (which is typically at least quadratic in the matter fields) and make use of the null energy condition.
    
    The equation of motion also splits into
    \begin{equation}
        \bm E_{ab} = \bm H_{ab} - \bm T_{ab}
    \end{equation}
    so that we can define the increasing entropy using solely the gravitational sector:
    \begin{equation}
        \partial_v^2 \delta S_\text{inc} = - 2 \pi \int_{\mathcal{C}(v)} \dd[D-2]{x} \sqrt{h}\, \delta H_{vv} = - 2 \pi \int_{\mathcal{C}(v)} \dd[D-2]{x} \sqrt{h}\, \delta T_{vv} \leq 0
    \end{equation}
    where we've used the null energy condition of the matter fields on-shell.  

    The covariant phase space description follows straightforwardly as the pre-symplectic potential $\bm \Theta$ and the Noether charge $\bm Q$ both split into gravitational and matter sectors
    \begin{equation}
        \bm \Theta = \bm \Theta^\text{g} + \bm \Theta^\text{m}, \quad \bm Q = \bm Q^\text{g} + \bm Q^\text{m}.
    \end{equation}
    
    In a covariant expression, the entropy is now defined in this case via
    \begin{equation}
        \pounds_k \pounds_k \delta S_\text{inc}(v) = 2 \pi \int_{\mathcal{C}(v)} \delta \left(\pounds_k \bm Q^\text{g}_k - \iota_k \bm \Theta^\text{g}_k\right)
    \end{equation}
    and the null energy condition ensures the entropy $S_\text{inc}$ to be strictly non-decreasing at any time-slice of the horizon
    \begin{equation}
        \pounds_k \delta S_\text{inc} \geq 0 \quad \text{for all }v
    \end{equation}
    when a teleological boundary condition $\pounds_k \delta S_\text{inc} \to 0$ at $v \to \infty$ is imposed, which follows if, e.g., the black hole settles down to a stationary state in the far future.
    
    \subsection{JKM Invariance of the Entropy}\label{sec:jkm}

    The Jacobson-Kang-Myers (JKM) ambiguities were originally discussed in papers \cite{Jacobson:1993vj} and \cite{Iyer:1994ys}. They consist of three different types of differential form ambiguities. Here, we give a quick review of these ambiguities and see how they would impact the covariant phase space quantities. Consider the following three classes of ambiguity that can arise at different levels 
    \begin{equation}
        \begin{aligned}
            \bm L[\phi] & \to \bm L[\phi] + \dd{\bm \alpha[\phi]},\\
            \bm \Theta[\phi, \delta \phi] & \to \bm \Theta[\phi, \delta \phi] + \dd{\bm \beta[\phi, \delta \phi]},\\
            \bm Q_\zeta & \to \bm Q_\zeta + \dd{\bm \gamma[\phi, \zeta]}.
        \end{aligned}
    \end{equation}
    
    The $\bm \beta, \bm \gamma$ ambiguities are required to have the following properties to match $\bm \Theta$ and $\bm Q$: $\bm \beta[\phi, \delta \phi]$ should be linear in $\delta \phi$, and $\bm \gamma[\phi, \zeta]$ should also be linear in $\zeta$.
    
    Under the above ambiguities, the following covariant phase space quantities transform as
    \begin{equation}
        \begin{split}
            \bm C_{\zeta} & \to \bm C_\zeta,\\
            \bm \Theta[\phi, \delta \phi] & \to \bm \Theta[\phi, \delta \phi] + \delta \bm \alpha[\phi] + \dd{\bm \beta[\phi, \delta \phi]},\\
            \bm Q_\zeta & \to \bm Q_\zeta + \iota_\zeta \bm \alpha[\phi] + \bm \beta[\phi, \pounds_\zeta \phi] + \dd{\bm \gamma[\phi, \zeta]}.
        \end{split}
    \end{equation}

    The constraint form $\bm C_\zeta$ is manifestly JKM-invariant as equations of motion should never be ambiguous. Hence, we expect our increasing entropy \eqref{eq:Sinc-cps-1} defined using $\delta \bm C_k$ to be invariant with respect to the JKM ambiguities. This can be verified by considering the following transformation
    \begin{equation}
        \left(\pounds_k \bm Q_k - \iota_k \bm \Theta_k \right) \to \left(\pounds_k \bm Q_k - \iota_k \bm \Theta_k \right) + \dd{\left(\iota_k \bm \beta[\phi, \pounds_k \phi] + \pounds_k \bm \gamma[\phi, k]\right)}
    \end{equation}
    where we've used $[\pounds_k, \iota_k] = 0$, $[\pounds_k, \dd{}] = 0$, and the Cartan-Killing formula. The only ambiguity left, which is the exact term, would contribute to the codimension-2 total derivative of the entropy current discovered in \cite{Bhattacharyya:2021jhr}. The increasing entropy itself is intact when the horizon slice is compact, as the exact term integrates out. Therefore, our increasing entropy is unambiguous.
    
    There is a purely semantic difference in how the JKM ambiguities were related to the increasing entropy, in the previous work of one of the authors \cite{Wall:2015raa}, when addressing positive boost-weight corrections to Wald entropy. In that paper (as reviewed in Section \ref{Noether_review}), the equations of motion were used to resolve the ambiguities in Wald entropy away from stationarity and away from the bifurcation surface. This was understood as ``fixing'' the entropy within the space of JKM ambiguities.  However, from our new point of view, this ambiguity ``fixing'' was actually an implicit recovery of the correction to Wald entropy from $\delta \bm Q_k$ and $\iota_k \delta \bm \Theta_k$.  In our current viewpoint, equation \eqref{eq:Sinc-cps-1} gives an unambiguous definition of the dynamical black hole entropy when two Lie derivatives are pulled out.  Rather than determining an ambiguous coefficient, the quantity we are interested in is simply not ambiguous in the first place.  However, the actual formula for the entropy is the same as in \cite{Wall:2015raa}, so this is simply a difference in the language used to describe the result.

    \subsection{Relation to Boost Noether Charge}
    The relationship between the increasing entropy and the \emph{boost} Noether charge $\bm Q_\xi$ is first pointed out by \cite{Hollands:2024vbe} as
    \begin{equation}\label{eq:Sinc-Snew}
        (1 - \kappa^{-1} \pounds_\xi) \delta S_\text{inc} = \frac{2 \pi}{\kappa}\int_{\mathcal{C}(v)} \left(\delta \bm Q_\xi - \iota_\xi \bm \Theta[\phi, \delta \phi]\right)
    \end{equation}
    where $\kappa$ is the surface gravity (we've implicitly normalised it to be 1 in this paper when constructing the GNC)\footnote{The behaviour of perturbations to the surface gravity is investigated in \cite{Visser:2024pwz}, and it would not contribute at first order of perturbation. Thus we are free to normalise the background surface gravity to be $1$ for convenience.}, and $\bm \Theta[\phi,\delta\phi]$ is the pre-symplectic potential as before. Its various aspects are also studied in \cite{Rignon-Bret:2023fjq,Visser:2024pwz,Ciambelli:2023mir}. In \cite{Hollands:2024vbe}, it is proven for pure gravity theories that the latter term $\iota_\xi \bm \Theta[\phi, \delta \phi]$ is exact in $\delta$ on the horizon, i.e., there is some form $\bm B$ such that $\delta(\iota_\xi \bm B[\phi]) = \iota_\xi \bm \Theta[\phi, \delta \phi]$, so altogether one can define an \emph{improved Noether charge} $\tilde{\bm Q}_\xi = \bm Q_\xi - \iota_\xi \bm B$ and then express the \emph{dynamical black hole entropy} (an alternative definition to our increasing entropy) in terms of it. The result is subsequently generalised to all non-minimally coupled bosonic matter fields in \cite{Visser:2024pwz}. These both provide further confirmations of the exactness of zero-boost terms. An easy corollary of the above relation is that the increasing entropy is equal to Wald entropy at the bifurcation surface, where $\xi$ vanishes.

    \section{Inauspicious Terms Caused by Vector Fields and Resolution}\label{sec:prob-term}

    In this section, we prove that the inauspicious terms involving $\delta V_v$ introduced by non-gauge vector fields will only contribute to the codimension-2 divergence of the entropy current, so the two $v$ indices in $\int \delta E_{vv}$ both correspond to $\partial_v$ rather than $\delta V_v$. Hence, the original prescription (pulling out two $\partial_v$'s) of the increasing entropy in \cite{Wall:2015raa} safely generalises to vector fields. We will carry out the proof using both the covariant phase space identity \eqref{eq:magic-id} and the structure analysis of $E_{vv}$ in Gaussian null coordinates.

    \subsection{Structure of $\delta E_{vv}$ without Vector Fields}
    We start with the structure of $\delta E_{vv}$. Using the results in Corollary \ref{cor:5} (also in \cite{Wall:2015raa,Bhattacharyya:2021jhr}), we can show, without a vector field in the theory, the structure of the variation of equation of motion is 
    \begin{equation}
        \delta E_{vv}^{(\text{no }V)} = \frac{1}{\sqrt{h}} \partial_v^2 \left(\sqrt{h} \sum_I \sum_{k\geq 0} A_{(-k)}^I \delta B_{(k)}^I\right) + D_i \tilde J^i \label{eq:Evv-no-V}
    \end{equation}
    where $A_{(k)}^I$, etc., are factors with Killing weight $k$, and the index $I$ keeps track of different terms with some specific Killing weight combination. In order to keep track of the entropy current in the spirit of \cite{Bhattacharyya:2021jhr}, we have separated out the terms $D_i \tilde J^i$ that are manifestly codimension-2 divergences, even though they can also be written in the $A \delta B$ form of the first term.  $\tilde J^i$ will receive corrections from the zero-boost terms, as discussed in Appendix \ref{app:zbt}.

    Here, defining the entropy via 
    \begin{equation}
        \partial_v^2 \delta S_\text{inc} \propto \int \dd[D-2]{x} \sqrt{h}\, \delta E_{vv}
    \end{equation}
    is not a problem, as we can safely take out two $v$-derivatives when integrating $\delta E_{vv}$ on a compact horizon time-slice.

    \subsection{Potential Problem of Vector Fields and Resolution}
    Adding a vector field causes a potential problem: there can be inauspicious terms with only one derivative, and the above definition may be invalid. Here, we aim to show that inauspicious terms sum to a codimension-2 divergence, which integrates out on a compact horizon slice.

    With the presence of a vector field, the structure of $\delta E_{vv}$ becomes 
    \begin{equation}
        \delta E_{vv} = \frac{1}{\sqrt{h}} \partial_v^2 \left(\sqrt{h} \sum_I\sum_{k\geq0} A_{(-k)}^I \delta B_{(k)}^I\right) + \frac{1}{\sqrt{h}} \partial_v \left(\sqrt{h}\, \mathcal P\right) + D_i \tilde J^i \label{eq:4-3}
    \end{equation}
    where 
    \begin{equation}
        \mathcal P = \sum_{I=0}^{n} (C_{(0)}^I)^{i_1 \cdots i_I} D_{(i_1} \cdots D_{i_I)} \delta V_v \label{eq:problem}
    \end{equation}
    are the inauspicious terms. We see such terms arise as by-products of vector fields, comparing with equation (\ref{eq:Evv-no-V}). Here, $n$ is the maximum number of $D_i$'s that show up in such terms, and all $D_i$ derivatives are symmetrised because of the structure of the theory. $C^I_{(0)}$ are weight-0 background quantities determined by the theory. In $\mathcal P$, there can be no $v$-derivative acting on $\delta V_v$ because if so we could remove it by adding a total $v$-derivative term, which would not be problematic.

    We resolve the problem of vectors by the following proposition, which is one of the key results of our paper.

    \begin{clm}
        The inauspicious term $\mathcal{P}$ is a codimension-2 total derivative.
    \end{clm}

    \begin{proof}
    We evaluate (\ref{eq:magic-id}) with $\zeta = \xi = v \partial_v - u \partial_u$ on the \emph{whole} horizon $\mathcal{H}$ from $v \to -\infty$ to $v \to +\infty$ in GNC on a stationary background. Then we obtain
    \begin{equation}
        \int_\mathcal{H} \delta \bm C_{\xi} = \left(\int_{\mathcal{C}(+\infty)} - \int_{\mathcal{C}(-\infty)} \right) \left(\delta \bm Q_\xi - \iota_\xi \bm \Theta[g, \delta g; V, \delta V]\right) \label{eq:red-magic-id}
    \end{equation}
    where $\bm \omega = 0$ as $\pounds_\xi g = \pounds_\xi V = 0$ on a stationary background, and $\xi$ is tangent to $\mathcal{H}$, so $\iota_\xi (\cdots) = 0$ on $\mathcal{H}$.

    Since we are off-shell, we can assume that the perturbations are \emph{compactly supported}, more specifically $\delta g, \delta V \to 0$ as $v \to \pm \infty$, hence the RHS of (\ref{eq:red-magic-id}) vanishes.

    Evaluating the other side,
    \begin{equation}
    \begin{split}
            \text{LHS of (\ref{eq:red-magic-id})} & = - \int_{v=-\infty}^{+\infty} \dd{v} \int_{\mathcal{C}(v)} \dd[D-2]{x} \sqrt{h}\, v\, \delta E_{vv}\\
            & = - \underbrace{\left[\int_{\mathcal{C}(v)} \dd[D-2]{x} \sqrt{h} \cdots\right]_{v=-\infty}^{+\infty}}_{0} + \int_{v=-\infty}^{+\infty} \dd{v} \int_{\mathcal{C}(v)} \dd[D-2]{x} \sqrt{h}\, \mathcal P\\
            & = \int_{v=-\infty}^{+\infty} \dd{v} \int_{\mathcal{C}(v)} \dd[D-2]{x} \sqrt{h}\, \mathcal P = 0
        \end{split}
    \end{equation}
    via integrating by parts.

    We proceed by plugging in the detailed inauspicious terms in (\ref{eq:problem}) to get
    \begin{equation}
    \begin{split}
            0 & = \sum_{I=0}^{n}\int_{v=-\infty}^{+\infty} \dd{v} \int_{\mathcal{C}(v)} \dd[D-2]{x} \sqrt{h}\, (C_{(0)}^I)^{i_1 \cdots i_I} D_{(i_1} \cdots D_{i_I)} \delta V_v\\
            & = \int_{v=-\infty}^{+\infty} \dd{v} \int_{\mathcal{C}(v)} \dd[D-2]{x} \sqrt{h}\, \delta V_v \sum_{I=0}^{n} (-1)^{I} D_{(i_1} \cdots D_{i_I)}(C_{(0)}^I)^{i_1 \cdots i_I}
        \end{split}
    \end{equation}
    which is obtained by performing integration by parts in the $D-2$ spatial directions for many times. We get an Euler-Lagrange-like equation 
    \begin{equation}
        C_{(0)}^0 - D_{i_1} (C_{(0)}^1)^{i_1} + D_{(i_1} D_{i_2)} (C_{(0)}^2)^{i_1 i_2} + \cdots + (-1)^n D_{(i_1} \cdots D_{i_n)} (C_{(0)}^n)^{i_1 \cdots i_n} = 0
    \end{equation}
    for the coefficients in $\mathcal{P}$.

    Now substituting this back we get
    \begin{equation}
    \begin{split}
            \mathcal{P} & = \sum_{I=1}^{n}(C_{(0)}^I)^{i_1 \cdots i_I} D_{(i_1} \cdots D_{i_I)}\delta V_v + (-1)^{I-1} \left(D_{(i_1} \cdots D_{i_I)} (C_{(0)}^I)^{i_1 \cdots i_I}\right) \delta V_v\\
            & = \sum_{I=1}^{n} \sum_{J=0}^{I-1} D_{i} \left((-1)^{J}D_{i_1} \cdots D_{i_J} (C^{I}_{(0)})^{(i i_1 \cdots i_{I-1})} D_{i_{J+1}} \cdots D_{i_{I-1}} \delta V_v\right)
        \end{split}
    \end{equation}
    through Leibniz rules.  Note that (inside the parentheses) in the side case $J=0$, there is no derivative on $C^I_{(0)}$; whereas in the side case $J=I-1$, there is no derivative on $\delta V_v$. 

    We conclude that 
    \begin{equation}
        \mathcal{P} = D_i J_V^i.
    \end{equation}
    where the entropy current from the inauspicious term is 
    \begin{equation}
        J_V^i = \sum_{I=1}^{n} \sum_{J=0}^{I-1} (-1)^{J}D_{i_1} \cdots D_{i_J} (C^{I}_{(0)})^{(i i_1 \cdots i_{I-1})} D_{i_{J+1}} \cdots D_{i_{I-1}} \delta V_v
    \end{equation}
    To close the proof, we note that the $(C_{(0)}^I)^{i_1 \cdots i_I}$ coefficients are fully specified by the theory and independent of the perturbations.  Furthermore, there is no ambiguity from removing the $\dd{v}$ integral since total $v$-derivatives do not appear in $\cal P$.  Therefore, the final structure of $\mathcal{P}$ that we demonstrated holds for arbitrary $\delta V_v$, i.e., it is valid for arbitrary non-compactly supported perturbations. 
    \end{proof}

    The inauspicious terms are resolved as the divergence integrates out on a compact horizon slice. It contributes to the entropy current, in the same spirit as the zero-boost corrections.

    \section{Example: One-Derivative Vector with $f(\text{Riemann})$}
    \label{sec:example1}
    Here, as an example, we carry out an explicit calculation for theories with only one derivative in vector fields non-minimally coupled to $f(\text{Riemann})$. To declutter, we choose to omit some minor details in the main text. These can be found in Appendix \ref{app:calc}. To provide a map for the readers, we have prepared Table \ref{tb:quantities} as a summary of quantities\footnote{The dynamical corrections are called ``JKM'' following the terminologies in \cite{Wall:2015raa}. See discussions in Section \ref{sec:jkm}.} appeared in the calculation and a flow diagram\footnote{In the flow diagram we've used $[\partial_v, D_i]=0,\, [\delta, D_i] T_{(w>0)} = 0,\, \partial_v \sqrt{h} = 0$ at first order of perturbation.} (Figure \ref{fig:flow-diagram}) for the whole calculation procedure. They are in the next page.

    \begin{table}[p]
        \centering
        \begin{tabular}{|c|c|}
            \hline
            Symbol & Description\\
            \hline
            $L_1$ & Lagrangian of one-derivative vector field coupled to $f(\text{Riemann})$\\
            $E_{ab}$ & equation of motion (EoM) of $g^{ab}$ from $L_1$\\
            $H_{ab}$ & gravitational contribution to EoM\\
            $\delta H_{\text{ZB}}$ & zero-boost corrections in $\delta H_{ab}$\\
            $\mathcal{V}_{ab}$ & vector contribution to EoM\\
            $\delta \mathcal{V}_{\text{ZB}}$ & zero-boost corrections in $\delta \mathcal{V}_{ab}$\\
            $\delta E_{\text{ZB}}$ & sum of $\delta H_{\text{ZB}}$ and $\delta \mathcal{V}_{\text{ZB}}$\\
            $\varsigma$ & entropy density\\
            $\varsigma_H, \varsigma_{\mathcal{V}}$ & $H,\mathcal{V}$ contributions to entropy density\\
            $\varsigma_{\text{Wald}}$ & Wald entropy density\\
            $\varsigma_{\text{JKM}}$ & dynamical corrections to Wald entropy density\\
            $J^i$ & entropy current\\
            $J_{H \text{M}}^i,J_{\mathcal{V}\text{M}}^i, J_{\text{M}}^i$ & manifest entropy current in $H,\mathcal{V}$ and their sum\\
            $J_{\text{ZB}}^i$ & $\delta E_{\text{ZB}}$ contributions to entropy current\\
            $S_\text{inc}$ & increasing black hole entropy\\
            \hline
        \end{tabular}
        \caption{Summary of quantities appeared in the calculation.} \label{tb:quantities}
    \end{table}

    \begin{figure}[p]
        \centering

\tikzset{every picture/.style={line width=0.75pt}} 

\begin{tikzpicture}[x=0.75pt,y=0.75pt,yscale=-1,xscale=1]

\draw (326.22,15.2) node    {$L_{1}$};
\draw (326.22,55.2) node    {$E_{ab} =H_{ab} +\mathcal{V}_{ab}$};
\draw (325.22,97.2) node    {$\delta E_{vv} =\delta H_{vv} +\delta \mathcal{V}_{vv}$};
\draw (167.56,152.87) node    {$\displaystyle \delta H_{vv} =\frac{1}{\sqrt{h}} \partial _{v}^{2} \delta \left(\sqrt{h}\, \varsigma _{H}\right) +D_{i}\left( \partial _{v} \delta J_{H\text{M}}^{i}\right) +\delta H_{\text{ZB}}$};
\draw (490.56,153.2) node    {$\displaystyle \delta \mathcal{V}_{vv} =\frac{1}{\sqrt{h}} \partial _{v}^{2} \delta \left(\sqrt{h}\, \varsigma _{\mathcal{V}}\right) +D_{i}\left( \partial _{v} \delta J_{\mathcal{V}\text{M}}^{i}\right) +\delta \mathcal{V}_{\text{ZB}}$};
\draw (325.89,257.2) node    {$\begin{cases}
\varsigma =\varsigma _{H} +\varsigma _{\mathcal{V}} =\varsigma _{\text{Wald}} +\varsigma _{\text{JKM}} & \\
\delta E_{\text{ZB}} =\delta H_{\text{ZB}} +\delta \mathcal{V}_{\text{ZB}} =D_{i}\left( \partial _{v} \delta J_{\text{ZB}}^{i}\right) & \\
J_{\text{M}}^{i} =J_{H\text{M}}^{i} +J_{\mathcal{V}\text{M}}^{i} & \\
J^{i} =J_{\text{M}}^{i} +J_{\text{ZB}}^{i} & 
\end{cases}$};
\draw (327.89,366.2) node    {$\displaystyle \delta E_{vv} =\frac{1}{\sqrt{h}} \partial _{v}^{2} \delta \left(\sqrt{h}\,( \varsigma _{\text{Wald}} +\varsigma _{\text{JKM}})\right) +D_{i}\left( \partial _{v} \delta J^{i}\right)$};
\draw (326.56,431.2) node    {$\displaystyle S_{\text{inc}} = - 2 \pi \int \mathrm{d}^{D-2} x\, \sqrt{h}\,( \varsigma _{\text{Wald}} +\varsigma _{\text{JKM}})$};
\draw    (326.22,28.2) -- (326.22,40.2) ;
\draw [shift={(326.22,42.2)}, rotate = 270] [color={rgb, 255:red, 0; green, 0; blue, 0 }  ][line width=0.75]    (10.93,-3.29) .. controls (6.95,-1.4) and (3.31,-0.3) .. (0,0) .. controls (3.31,0.3) and (6.95,1.4) .. (10.93,3.29)   ;
\draw    (325.91,68.2) -- (325.58,82.2) ;
\draw [shift={(325.53,84.2)}, rotate = 271.36] [color={rgb, 255:red, 0; green, 0; blue, 0 }  ][line width=0.75]    (10.93,-3.29) .. controls (6.95,-1.4) and (3.31,-0.3) .. (0,0) .. controls (3.31,0.3) and (6.95,1.4) .. (10.93,3.29)   ;
\draw    (288.4,110.2) -- (248.75,124.2) ;
\draw [shift={(246.86,124.87)}, rotate = 340.55] [color={rgb, 255:red, 0; green, 0; blue, 0 }  ][line width=0.75]    (10.93,-3.29) .. controls (6.95,-1.4) and (3.31,-0.3) .. (0,0) .. controls (3.31,0.3) and (6.95,1.4) .. (10.93,3.29)   ;
\draw    (363.6,110.2) -- (406,124.56) ;
\draw [shift={(407.89,125.2)}, rotate = 198.71] [color={rgb, 255:red, 0; green, 0; blue, 0 }  ][line width=0.75]    (10.93,-3.29) .. controls (6.95,-1.4) and (3.31,-0.3) .. (0,0) .. controls (3.31,0.3) and (6.95,1.4) .. (10.93,3.29)   ;
\draw    (446.22,181.2) -- (420.21,197.63) ;
\draw [shift={(418.51,198.7)}, rotate = 327.72] [color={rgb, 255:red, 0; green, 0; blue, 0 }  ][line width=0.75]    (10.93,-3.29) .. controls (6.95,-1.4) and (3.31,-0.3) .. (0,0) .. controls (3.31,0.3) and (6.95,1.4) .. (10.93,3.29)   ;
\draw    (210.05,180.87) -- (235.44,197.6) ;
\draw [shift={(237.11,198.7)}, rotate = 213.38] [color={rgb, 255:red, 0; green, 0; blue, 0 }  ][line width=0.75]    (10.93,-3.29) .. controls (6.95,-1.4) and (3.31,-0.3) .. (0,0) .. controls (3.31,0.3) and (6.95,1.4) .. (10.93,3.29)   ;
\draw    (326.96,315.7) -- (327.34,336.2) ;
\draw [shift={(327.38,338.2)}, rotate = 268.95] [color={rgb, 255:red, 0; green, 0; blue, 0 }  ][line width=0.75]    (10.93,-3.29) .. controls (6.95,-1.4) and (3.31,-0.3) .. (0,0) .. controls (3.31,0.3) and (6.95,1.4) .. (10.93,3.29)   ;
\draw    (327.31,394.2) -- (327,409.7) ;
\draw [shift={(326.96,411.7)}, rotate = 271.18] [color={rgb, 255:red, 0; green, 0; blue, 0 }  ][line width=0.75]    (10.93,-3.29) .. controls (6.95,-1.4) and (3.31,-0.3) .. (0,0) .. controls (3.31,0.3) and (6.95,1.4) .. (10.93,3.29)   ;

\end{tikzpicture}
        \caption{Flow diagram of the calculation.} \label{fig:flow-diagram}
    \end{figure}
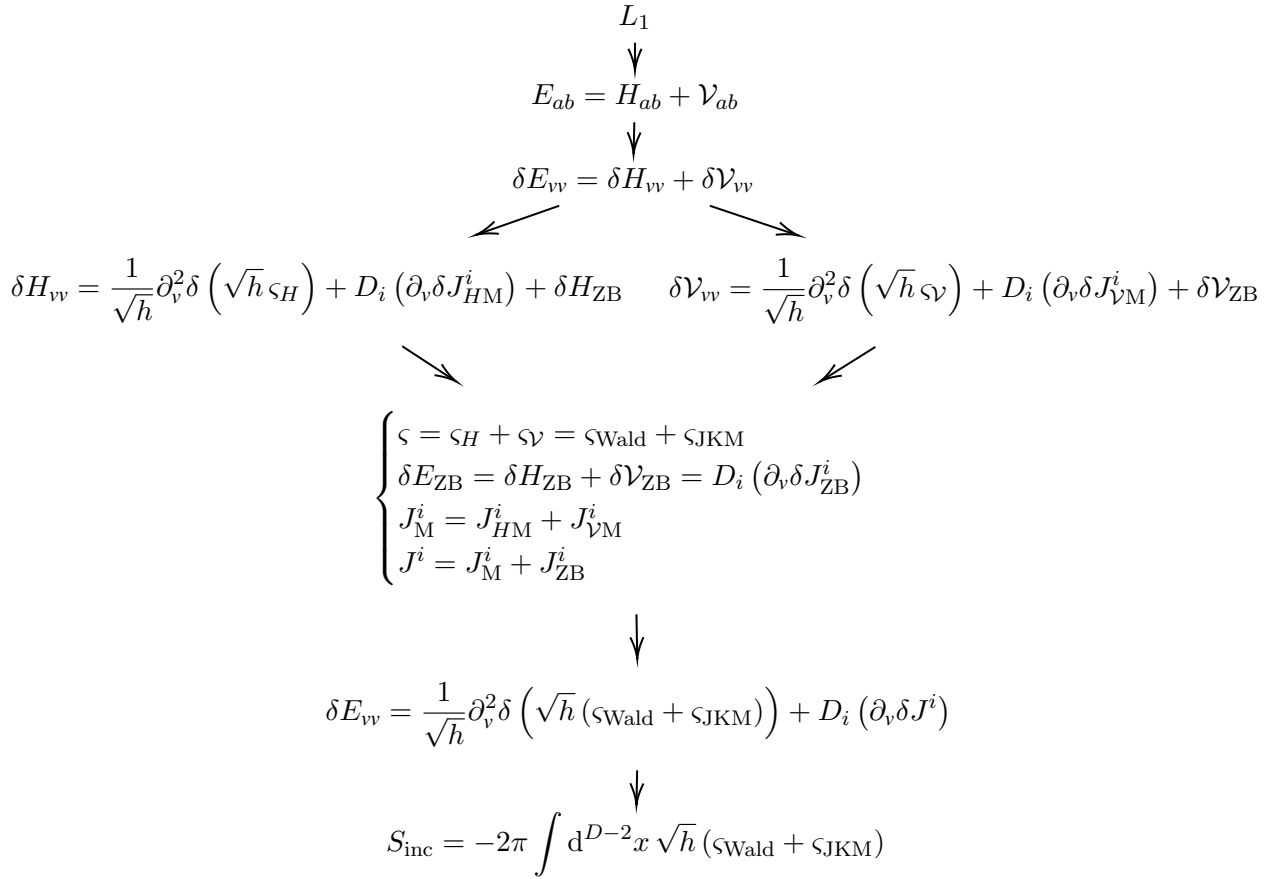

    For readers who are eager to see the results quickly, you can skip the calculations and read Section \ref{sec:vec-formula} directly.

    \subsection{Equations of Motion}
    The theory we are interested in has an action
    \begin{equation}
        I_1 = \int \dd[D]{x} \sqrt{-g}\, L_1
    \end{equation}
    with Lagrangian
    \begin{equation}
        L_1 = L_1(g^{ab}, R_{abcd}, V_a, \nabla_b V_a),
    \end{equation}
    which is a locally diffeomorphism invariant scalar constructed by linear combinations of contractions of the arguments. All indices should be contracted, and no extra derivatives are allowed to act on these arguments. 
    
    To find the equation of motions, we vary the action 
    \begin{equation}
        \delta (\sqrt{-g} L_1) = \sqrt{-g} \left(X^{abcd} \delta R_{abcd} + \left(Y_{ab} - \frac{1}{2} g_{ab} L_1\right) \delta g^{ab} + Z^a \delta V_a + Z^{a|b} \delta(\nabla_b V_a)\right)
            \label{eq:L1-var}
    \end{equation}
    where we define
    \begin{equation}
        X^{abcd} = \pdv{L_1}{R_{abcd}}, \quad Y_{ab} = \pdv{L_1}{g^{ab}}, \quad Z^a = \pdv{L_1}{V_a}, \quad Z^{a|b} = \pdv{L_1}{(\nabla_b V_a)}.
    \end{equation}

    For later, we calculate the standard variation of Riemann tensor
    \begin{equation}
        \delta R_{abcd} = \frac{1}{2} \left(\nabla_{b}\nabla_{c} \delta g_{ad} - \nabla_{b}\nabla_{d} \delta g_{ac} - \nabla_{a} \nabla_{c} \delta g_{bd} + \nabla_{a} \nabla_{d} \delta g_{bc}\right) + R\indices{^e _{bcd}}\delta g _{ae},
    \end{equation}
    and we notice that $\delta$ and $\nabla$ on the vector field does not commute
    \begin{equation}
        \delta(\nabla_b V_a) = \nabla_b \delta V_a - (\delta \conn{c}{b}{a}) V_c.
    \end{equation}

    Setting the boundary condition such that the total derivative term integrate to zero, we obtain the Euler-Lagrange equation of $V_a$
    \begin{equation}
        \mathcal{E}^a = Z^a - \nabla_b Z^{a|b} = 0 \label{eq:v-eom}
    \end{equation}
    and the equation of motion of $g^{ab}$ 
    \begin{equation}
        \frac{1}{2} E_{ab} = 2 \nabla^{(c} \nabla^{d)} X_{acbd} - X\indices{_{(a}^{cde}} R_{b)cde} + Y_{ab} - \frac{1}{2} g_{ab} L_1 + \frac{1}{2} \nabla^c \left(Z_{(a|b)} V_c - Z_{(a|c} V_{b)} - Z_{c|(a} V_{b)} \right) = 0
    \end{equation}
    with the factor of half a convention, which is mentioned earlier and also later.

    To further simplify, we find the relationship between $Y_{ab}$ and $X^{abcd}, Z^a, Z^{a|b}$ using a generalisation of method in \cite{Padmanabhan:2011ex}:
    \begin{equation}
        Y_{ab} = 2 X\indices{_{(a}^{cde}}R_{b)cde} + \frac{1}{2} (Z_{(a} V_{b)} + Z_{(a|c}\nabla^c V_{b)} + Z_{c|(a}\nabla_{b)} V^c), \label{eq:Y_ab}
    \end{equation}
    and the details are presented in Appendix \ref{app:calc}.

    Finally, the equation of motion of $g^{ab}$ is 
    \begin{equation}
    \begin{split}
            E_{ab} & = 4\nabla^{(c} \nabla^{d)} X_{acbd} + 2 X\indices{_{(a}^{cde}}R_{b)cde} - g_{ab} L_1\\
            & \quad + \mathcal{E}_{(a} V_{b)}+ Z_{c|(a} \nabla_{b)} V^c + \nabla^c \left(Z_{(a|b)}V_c - Z_{c|(a}V_{b)}\right) = 0
        \end{split}
    \end{equation}
    where we followed the convention of \cite{Wall:2015raa} that 
    \begin{equation}
        E_{ab} = \frac{2}{\sqrt{-g}} \frac{\delta I}{\delta g^{ab}}.
    \end{equation}

    For our analysis later, let's write the \emph{off-shell} equation of motion tensor as
    \begin{equation}
        E_{ab} = H_{ab} + \mathcal{V}_{ab}
    \end{equation}
    where $H_{ab}$ is the \emph{gravitational contribution}, keeping the same functional form as the EoM of $f(\text{Riemann})$ theory\footnote{Caveat: $H_{ab}$ also depends on the vector field \emph{implicitly} through the $X_{abcd}$ tensor and the Lagrangian $L_1$.}
    \begin{equation}
        H_{ab} = 4\nabla^{(c} \nabla^{d)} X_{acbd} + 2 X\indices{_{(a}^{cde}}R_{b)cde} - g_{ab} L_1
    \end{equation}
    and $\mathcal{V}_{ab}$ is the extra \emph{vector contribution}\footnote{These are the \emph{explicit} vector contribution to the equation of motion.} 
    \begin{equation}
        \mathcal{V}_{ab} = \mathcal{E}_{(a} V_{b)} + Z_{c|(a} \nabla_{b)} V^c + \nabla^c \left(Z_{(a|b)}V_c - Z_{c|(a}V_{b)}\right).
    \end{equation}
    
    \subsection{Off-Shell Analysis of $\delta E_{vv}$}
    The entropy formula we are calculating is determined by the off-shell structure of the equation of motion tensor. Here, we present a detailed analysis.

    In the case of linear perturbation of a stationary background $g = \bar g + \delta g$, the variation of the $vv$-component of the equation of motion evaluated on the horizon is
    \begin{equation}
    \begin{split}
            \delta E_{vv} & = \delta H_{vv} + \delta \mathcal{V}_{vv}\\
            & = \delta \left(4\nabla^{(c} \nabla^{d)} X_{vcvd} + 2 (X\indices{_v^{cde}}R_{vcde})\right) + \delta \left(\mathcal{E}_v V_v + Z_{c|v} \nabla_v V^c + \nabla^c \left(Z_{v|v}V_c - Z_{c|v}V_v\right)\right)
        \end{split}
    \end{equation}
    where we used $g_{vv} = 0, g_{uu} = 0, g_{uv} = 1, g_{ui} = g_{vi} = 0$ and the appropriate inverse on the horizon. We analyse this term by term. 
    \subsubsection{Contributions from $\delta H_{vv}$}
    We decompose
    \begin{equation}
        \delta H_{vv} = \delta H^{(1)} + \delta H^{(2)}
    \end{equation}
    with
    \begin{equation}
        \delta H^{(1)} = 4 \delta (\nabla_v \nabla_a X^{uvua}) + 4 \delta (\nabla_i \nabla_a X^{uiua}),
    \end{equation}
    and 
    \begin{equation}
            \delta H^{(2)} = 2 \delta (X\indices{_v ^{cde}} R_{vcde}) = 4 (R_{viuj} \delta X\indices{_v^{iuj}} + X\indices{_v^{ivj}} \delta R_{vivj}), \label{eq:H-2}
    \end{equation}
    which is a sum of zero-boost terms (see Appendix \ref{app:zbt} for the definition). Here, we used the result in Claim \ref{clm:2} that a positive weight background tensor component vanishes. The statement in Claim \ref{clm:2} will be repeatedly used throughout to simplify the calculations.

    Back to $\delta H^{(1)}$, we expand the covariant derivatives by reading off Table \ref{tb:covder}
    \begin{equation}
        \delta H^{(1)} = 4 \delta \left(\partial_v (\nabla_a X^{uvua}) + \frac{1}{2} \omega_i \nabla_a X^{uiua} + D_i(\nabla_a X^{uiua})\right).
    \end{equation}

    Here we found that the second and the third terms are a zero-boost term and a codimension-2 total derivative, respectively. We would postpone the analysis of zero-boost terms after collecting all of them.

    We further expand the $v$-derivative term (which is the first term above),
    \begin{equation}
    \begin{split}
            \delta H^ {(1,1)} & = 4 \delta (\partial_v (\nabla_a X^{uvua}))\\
            & = 4 \biggl(\frac{1}{\sqrt{h}}\partial_v^2 \delta \left(\sqrt{h}\, X^{uvuv}\right) - \frac{1}{2} \omega_i \partial_v \delta X^{uvui} +  D_i (\partial_v \delta X^{uvui})\\
            & \qquad \quad +  X^{uivj} \partial_v \delta K_{ij} - \partial_v (\bar K_{ij}) \delta X^{uiuj} -  \bar K_{ij} \partial_v \delta X^{uiuj} \biggr)
        \end{split}\label{eq:H-1-1}
    \end{equation}
    where we identified
    \begin{equation}
        \partial_v (\delta D_i X^{uvui}) = \partial_v (\partial_i \delta X^{uvui} + \delta (\conn{i}{i}{j} X^{uvuj})) = D_i(\partial_v \delta X^{uvui}), \label{eq:dv-delta-D-X}
    \end{equation}
    and we have recovered the factor of $\delta \sqrt{h}$ through
    \begin{equation}
        \partial_v \delta K = \frac{1}{\sqrt{h}} \partial_v^2 \delta \sqrt{h}.
    \end{equation}

    In summary, we can write 
    \begin{equation}
    \begin{split}
            \delta H_{vv} = \frac{4}{\sqrt{h}}\partial_v^2 \delta \left(\sqrt{h}\,X^{uvuv}\right) - 4 \bar K_{ij} \partial_v \delta X^{uiuj} + \delta H_{\text{pZB}} + D_i \left(\partial_v \delta J_{H \text{M}}^i\right)
        \end{split}
    \end{equation}
    where the first term is the derivative of Wald entropy density, $\delta H_{\text{pZB}}$ is a preliminary collection of zero-boost corrections to the derivative of Wald entropy density\footnote{The Wald entropy density itself includes zero-boost terms, i.e., the Iyer-Wald entropy density. Here we separate out the Wald contribution deliberately for later convenience.}
    \begin{equation}
    \begin{split}
            \delta H_{\text{pZB}} & = 4\biggl(R_{viuj}\delta X^{uiuj} + X^{uivj} \delta R_{vivj}+ \frac{1}{2}\omega_i \delta (\nabla_a X^{uiua})\\
            & \qquad \qquad \qquad \qquad \! - \frac{1}{2} \omega_i \partial_v \delta X^{uvui} + X^{uivj}\partial_v \delta K_{ij} - (\partial_v \bar K_{ij}) \delta X^{uiuj}\biggr),
        \end{split}
    \end{equation}
    which will be analysed later, and 
    \begin{equation}
        \partial_v \delta J_{H \text{M}}^i = 4 (\delta (\nabla_a X^{uiua}) + \partial_v \delta X^{uvui})
    \end{equation}
    is the manifest $H_{ab}$ contribution to the entropy current. The divergence term will not contribute when we integrate over the codimension-2 (compact) surface.

    Till this point, everything is the same as the pure $f(\text{Riemann})$ case as in \cite{Wall:2015raa}. Now, we will analyse the pre-JKM term (the ``pre'' refers to the fact that this will still contribute to the pre-zero-boost corrections when taking out two $v$-derivatives)
    \begin{equation}
        \delta H_{\text{pJKM}} = -4 \bar K_{ij} \partial_v \delta X^{uiuj} \label{eq:H-pJKM}
    \end{equation}
    and see what new contributions we have from the vector field. To find the details, we expand the variation using chain rule (by noticing $[\delta, \partial] = 0$)
    \begin{equation}
    \begin{split}
            \delta H_{\text{pJKM}} & = - 4 \bar K_{ij} \delta \left(\pdv{X^{uiuj}}{g^{ab}} \partial_v g^{ab} + \pdv{X^{uiuj}}{R_{abcd}} \partial_v R_{abcd} + \pdv{X^{uiuj}}{V_a} \partial_v V_a + \pdv{X^{uiuj}}{(\nabla_b V_a)} \partial_v(\nabla_b V_a) \right).
        \end{split}
    \end{equation}

    To filter out the valid components from all these, we write schematically
    \begin{equation}
        \underbrace{H_{\text{pJKM}}}_{\text{weight 2}} \quad \sim \underbrace{\bar K_{ij}}_{\text{weight }-1} \times \underbrace{\pdv{X^{uiuj}}{A}}_{\text{weight }(2-k)} \times \underbrace{\partial_v A}_{\text{weight }(1+k)}
    \end{equation}
    where $A \in \{g^{ab}, R_{abcd}, V_a, \nabla_b V_a\}$ is the component of interest. We observe that in this theory, weight of $A$ can be $k = -2,-1,0,1,2$. For $k = 0,1$, the term is manifestly vanishing as any product of two positive weight components is zero at first order, so we ask for $k = -2, -1, 2$. However, for $k = -2, -1$, the term $\pdv{X^{uiuj}}{A}$ is of weight $4,3$, respectively. But everything is constructed from products of components of $g^{ab}, R_{abcd}, V_a, \nabla_b V_a$, each factor can give at most weight $2$, so weight $4,3$ terms are of second order. Therefore, we are left with the only choice $k = 2$. Such statement should readily generalise to theories with higher order derivatives of $V_a$.

    Components satisfying $k = 2$ are 
    \begin{equation}
        g^{uu}, \quad R_{vkvl}, \quad \nabla_v V_v
    \end{equation}
    but $\partial_v g^{uu} = 0$ on the horizon by GNC conditions, so we only proceed with the other two. Taking into account the combinatorics factor $4$ for $\pdv{R_{vkvl}}$, we get 
    \begin{equation}
    \begin{split}
            \delta H_{\text{pJKM}} = - 4 \bar K_{ij} \left(4 \pdv{X^{uiuj}}{R_{vkvl}} \partial_v \delta R_{vkvl} + \pdv{X^{uiuj}}{(\nabla_v V_v)} \partial_v \delta (\nabla_v V_v)\right).
        \end{split}
    \end{equation}
    Evaluating on the horizon $u = 0$
	\begin{equation}
	    \delta R_{vkvl} = - \partial_v \delta K_{kl},
    \end{equation}
    and 
    \begin{equation}
        \delta (\nabla_v V_v) = \partial_v \delta V_v,
    \end{equation}
    we finally have
    \begin{equation}
        \delta H_{\text{pJKM}} = 16 \pdv{L_1}{R_{uiuj}}{R_{vkvl}}\bar K_{ij} \partial_v^2 \delta K_{kl} - 4 \pdv{L_1}{R_{uiuj}}{(\nabla_v V_v)} \bar K_{ij} \partial_v^2 \delta V_v.
    \end{equation}

    Using the product rule, we have
    \begin{equation}
    \begin{split}
            \delta H_{\text{pJKM}} & = \partial_v^2 \left(16 \pdv{L_1}{R_{uiuj}}{R_{vkvl}}\bar K_{ij} \delta K_{kl} - 4 \pdv{L_1}{R_{uiuj}}{(\nabla_v V_v)} \bar K_{ij} \delta V_v \right)\\
            & \quad + 8(\partial_v \bar K_{ij})\left(4 \pdv{L_1}{R_{uiuj}}{R_{vkvl}} \delta R_{vkvl} + \pdv{L_1}{R_{uiuj}}{(\nabla_v V_v)} \delta (\nabla_v V_v) \right)
        \end{split}
    \end{equation}
    where the second line is a zero-boost term.

    Hence we can conclude (by inserting factors of $\sqrt{h}, \delta \sqrt{h}$):
    \begin{equation}
        \boxed{\begin{aligned}
            \delta H_{vv} & = \frac{1}{\sqrt{h}}\partial_v^2 \delta \left(\sqrt{h}\left(4 \pdv{L_1}{R_{uvuv}} + 16 \pdv{L_1}{R_{uiuj}}{R_{vkvl}}\bar K_{ij} K_{kl} - 4 \pdv{L_1}{R_{uiuj}}{(\nabla_v V_v)} \bar K_{ij} V_v\right)\right)\\[0.5em]
            & \quad + D_i \left( \partial_v \delta J^i_{H \text{M}} \right) + \delta H_{\text{ZB}} 
        \end{aligned}}
    \end{equation}
    where
    \begin{equation}
        \begin{split}
            \delta H_{\text{ZB}} & = 4\biggl(R_{viuj}\delta X^{uiuj} + X^{uivj} \delta R_{vivj}+ \frac{1}{2}\omega_i \delta (\nabla_a X^{uiua})\\
            & \qquad \quad - \frac{1}{2} \omega_i \partial_v \delta X^{uvui} + X^{uivj}\partial_v \delta K_{ij} - (\partial_v \bar K_{ij}) \delta X^{uiuj}\\
            & \qquad \quad + 2(\partial_v \bar K_{ij})\left(4 \pdv{X^{uiuj}}{R_{vkvl}} \delta R_{vkvl} + \pdv{X^{uiuj}}{(\nabla_v V_v)} \delta (\nabla_v V_v) \right)\biggr),
        \end{split}
    \end{equation}
    is the true zero-boost term contributed by $\delta H_{vv}$, which will be analysed later.
    \subsubsection{Contributions from $\delta \mathcal{V}_{vv}$}
    Now, for the $\mathcal{V}_{vv}$ term,
    \begin{equation}
    \begin{split}
            \delta \mathcal{V}_{vv} & = \delta \left(\mathcal{E}_v V_v + Z_{c|v} \nabla_v V^c + \nabla^c \left(Z_{v|v}V_c - Z_{c|v}V_v\right)\right)\\
            & = \delta \mathcal{V}^{(1)} + \delta \mathcal{V}^{(2)}
        \end{split}
    \end{equation}
    where we've identified $\delta (\mathcal{E}_v V_v) = 0$, and 
    \begin{equation}
            \delta \mathcal{V}^{(1)} = Z_{u|v} \delta (\nabla_v V^u) + (\nabla_v V^v) \delta Z_{v|v} \label{eq:V-1}
    \end{equation}
    where the two terms are both zero-boost terms, and 
    \begin{equation}
            \delta \mathcal{V}^{(2)} = \delta \nabla^c \left(Z_{v|v}Z_c - Z_{c|v} V_v\right) = \delta \mathcal{V}^{(2,1)} + \delta \mathcal{V}^{(2,2)} + \delta \mathcal{V}^{(2,3)}
    \end{equation}
    where 
    \begin{align}
        \delta \mathcal{V}^{(2,1)} & = \delta \nabla_u (Z_{v|v} V_v - Z_{v|v} V_v) = 0,\\
        \delta \mathcal{V}^{(2,2)} & = \delta \nabla_v \left(Z_{v|v}V^v - Z^{v|u}V_v\right) = V_u \partial_v(\delta Z_{v|v}) + (\nabla_v V^v) \delta Z_{v|v} - Z_{u|v} \delta (\nabla_v V^u), \label{eq:V-2-2}\\
        \delta \mathcal{V}^{(2,3)} & = \delta \nabla_i \left(Z_{v|v}V^i - Z^{i|u}V_v\right) =  D_i (V^i \delta Z_{v|v} ) + \omega_i V^i \delta Z_{v|v}. \label{eq:V-2-3}
    \end{align}

    Then, altogether, we have 
    \begin{equation}
        \delta \mathcal{V}_{vv} = V_u \partial_v\delta Z_{v|v} + \delta \mathcal{V}_{\text{pZB}} + D_i \left( \partial_v \delta J^i_{\mathcal{V}\text{M}} \right)
    \end{equation}
    where $\delta \mathcal{V}_{\text{pZB}}$ are preliminary zero-boost corrections 
    \begin{equation}
        \delta \mathcal{V}_{\text{pZB}} = (2\nabla_v V^v + \omega_i V^i) \delta Z_{v|v} = 2 (\partial_v V_u) \delta Z^{u|u},
    \end{equation}
    which will be analysed later, and 
    \begin{equation}
        \partial_v \delta J_{\mathcal{V}\text{M}}^i = V^i \delta Z_{v|v}
    \end{equation}
    is the manifest $\mathcal{V}_{ab}$ contribution to the entropy current.

    Here, we have shown that the ``inauspicious terms'' vanish explicitly for one-derivative case
    \begin{equation}
        \partial_v \mathcal{P} = C_{(0)} \partial_v \delta V_v = 0.
    \end{equation}

    The first term of $\delta \mathcal{V}_{vv}$ is a pre-JKM term
    \begin{equation}
        \delta \mathcal{V}_{\text{pJKM}} = V_u \partial_v \delta Z_{v|v}
    \end{equation}
    and we now analyse it using similar arguments as for $\delta H_{\text{pJKM}}$. Using the chain rule 
    \begin{equation}
    \begin{split}
            \delta \mathcal{V}_{\text{pJKM}} & = V_u \delta \left(\pdv{Z_{v|v}}{g^{ab}} \partial_v g^{ab} + \pdv{Z_{v|v}}{R_{abcd}} \partial_v R_{abcd} + \pdv{Z_{{v|v}}}{V_a} \partial_v V_a + \pdv{Z_{v|v}}{(\nabla_b V_a)} \partial_v(\nabla_b V_a)\right)
        \end{split}
    \end{equation}
    and identifying the non-vanishing weight 2 components $\delta R_{vivj} = - \partial_v \delta K_{ij}$ and $\delta (\nabla_v V_v) = \partial_v \delta V_v$, we have 
    \begin{equation}
    \begin{split}
            \delta \mathcal{V}_{\text{pJKM}} & = - 4 \pdv{L_1}{(\nabla_u V_u)}{R_{vivj}} V_u \partial_v^2 \delta K_{ij} + \pdv{L_1}{(\nabla_u V_u)}{(\nabla_v V_v)} V_u \partial_v^2 \delta V_v\\
            & = \partial_v^2 \left(- 4 \pdv{L_1}{(\nabla_u V_u)}{R_{vivj}} V_u \delta K_{ij} + \pdv{L_1}{(\nabla_u V_u)}{(\nabla_v V_v)} V_u \delta V_v\right)\\
            & \quad - 2 (\partial_v V_u) \left(4\pdv{L_1}{(\nabla_u V_u)}{R_{vivj}} \delta R_{vivj} + \pdv{L_1}{(\nabla_u V_u)}{(\nabla_v V_v)} \delta (\nabla_v V_v)\right)
        \end{split}
    \end{equation}
    where we also raised the index of $\delta Z_{v|v}$ to $\delta Z^{u|u}$ considering the metric and Christoffel symbols on the horizon.

    Then we can conclude (also inserting $\sqrt{h}, \delta \sqrt{h}$):
    \begin{equation}
        \boxed{\begin{aligned}
            \delta \mathcal{V}_{vv} & = \frac{1}{\sqrt{h}}\partial_v^2 \delta \left(\sqrt{h} \left(- 4 \pdv{L_1}{(\nabla_u V_u)}{R_{vivj}} V_u K_{ij} + \pdv{L_1}{(\nabla_u V_u)}{(\nabla_v V_v)} V_u V_v\right)\right)\\[0.5em]
            & \quad + D_i \left( \partial_v \delta J^i_{\mathcal{V}\text{M}} \right) + \delta \mathcal{V}_{\text{ZB}}
        \end{aligned}}
    \end{equation}
    where 
    \begin{equation}
        \delta \mathcal{V}_{\text{ZB}} = 2 (\partial_v V_u) \left(\delta Z^{u|u}-4\pdv{Z^{u|u}}{R_{vivj}} \delta R_{vivj} - \pdv{Z^{u|u}}{(\nabla_v V_v)} \delta (\nabla_v V_v)\right)
    \end{equation}
    is the true zero-boost correction contributed by $\delta \mathcal{V}_{vv}$, which will also be analysed later.

    \subsubsection{Summary of $\delta E_{vv}$}
    Summarise all above, on the horizon, we have 
    \begin{equation}
        \boxed{\begin{aligned}
            \delta E_{vv} & = \frac{1}{\sqrt{h}}\partial_v^2 \delta \Biggl(\sqrt{h} \biggl(4 \pdv{L_1}{R_{uvuv}} + 16 \pdv{L_1}{R_{uiuj}}{R_{vkvl}}\bar K_{ij} K_{kl} - 4 \pdv{L_1}{R_{uiuj}}{(\nabla_v V_v)} \bar K_{ij} V_v\\
            & \qquad - 4 \pdv{L_1}{(\nabla_u V_u)}{R_{vivj}} V_u K_{ij} + \pdv{L_1}{(\nabla_u V_u)}{(\nabla_v V_v)} V_u V_v \biggr)\Biggr) + D_i \left( \partial_v \delta J_{\text{M}}^i\right) + \delta E_{\text{ZB}} 
        \end{aligned}}
    \end{equation}
    where 
    \begin{equation}
        \delta E_{\text{ZB}} = \delta H_{\text{ZB}} + \delta \mathcal{V}_{\text{ZB}} \quad \text{and} \quad J_{\text{M}}^i = J^i_{H \text{M}} + J^i_{\mathcal{V}\text{M}}
    \end{equation}
    with 
    \begin{equation}
    \begin{split}
            \delta H_{\text{ZB}} & = 4\biggl(R_{viuj}\delta X^{uiuj} + X^{uivj} \delta R_{vivj}+ \frac{1}{2}\omega_i \delta (\nabla_a X^{uiua})\\
            & \qquad - \frac{1}{2} \omega_i \partial_v \delta X^{uvui} + X^{uivj}\partial_v \delta K_{ij} - (\partial_v \bar K_{ij}) \delta X^{uiuj}\\
            & \qquad + 2(\partial_v \bar K_{ij})\left(4 \pdv{X^{uiuj}}{R_{vkvl}} \delta R_{vkvl} + \pdv{X^{uiuj}}{(\nabla_v V_v)} \delta (\nabla_v V_v) \right)\biggr),\\[0.5em]
            \delta \mathcal{V}_{\text{ZB}} & = 2 (\partial_v V_u) \left(\delta Z^{u|u}-4\pdv{Z^{u|u}}{R_{vivj}} \delta R_{vivj} - \pdv{Z^{u|u}}{(\nabla_v V_v)} \delta (\nabla_v V_v)\right)
        \end{split}
    \end{equation}
    and
    \begin{equation}
        \partial_v \delta J_{\text{M}}^i = 4 (\delta (\nabla_a X^{uiua}) + \partial_v \delta X^{uvui}) + V^i \delta Z_{v|v}.
    \end{equation}
    \subsection{Analysis of Zero-Boost Corrections}
    According to a general derivation in \cite{Bhattacharyya:2021jhr} (and a refined version in \cite{Hollands:2022fkn}, reviewed in Appendix \ref{app:zbt}), the zero-boost terms, after removing the Iyer-Wald part, can be identified as a codimension-2 spatial total derivative $\delta E_{\text{ZB}} = D_i \left( \partial_v \delta J_{\text{ZB}}^i\right)$ so that the total entropy current can be found as $J^i = J_{\text{M}}^i + J_{\text{ZB}}^i$. We now demonstrate this result explicitly through our calculations.

    First we look at the zero-boost term from $\delta \mathcal{V}_{vv}$:
    \begin{equation}
        \delta \mathcal{V}_{\text{ZB}} = 2 (\partial_v V_u) \left(\delta Z^{u|u} - 4\pdv{Z^{u|u}}{R_{vivj}} \delta R_{vivj} - \pdv{Z^{u|u}}{(\nabla_v V_v)} \delta (\nabla_v V_v)\right).
    \end{equation}

    Here, $\delta Z^{u|u}$ is of weight 2. We can expand it in terms of perturbations of the fields
    \begin{equation}
        \delta Z^{u|u} = \pdv{Z^{u|u}}{g^{ab}} \delta g^{ab} + \pdv{Z^{u|u}}{R_{abcd}} \delta R_{abcd} + \pdv{Z^{u|u}}{V_a} \delta V_a + \pdv{Z^{u|u}}{(\nabla_b V_a)} \delta (\nabla_b V_a).
    \end{equation}
    Now, if we consider the schematic form 
    \begin{equation}
        \underbrace{\delta Z^{u|u}}_{\text{weight 2}} \sim \underbrace{\pdv{Z^{u|u}}{A}}_{\text{weight }(2-k)} \times \underbrace{\delta A}_{\text{weight }k}
    \end{equation}
    for some component $A$ with weight $k$, we immediately see that the non-vanishing terms of $\mathcal{O}(\epsilon)$ are from $k \geq 2$. All other ones are zero as positive weight background components appear. Then on the horizon
    \begin{equation}
        \delta Z^{u|u} = 4\pdv{Z^{u|u}}{R_{vivj}} \delta R_{vivj} + \pdv{Z^{u|u}}{(\nabla_v V_v)} \delta (\nabla_v V_v)
    \end{equation}
    by identifying the corresponding symmetry factor $4$ of $\delta R_{vivj}$. Finally, we conclude 
    \begin{equation}
        \delta \mathcal{V}_{\text{ZB}} = 0.
    \end{equation}

    Then we look at the zero-boost term from $\delta H_{vv}$:
    \begin{equation}
    \begin{split}
            \delta H_{\text{ZB}} & = 4\biggl(\underbrace{R_{viuj} \delta X^{uiuj}}_{\delta H_{\text{ZB1}}} + \underbrace{\frac{1}{2}\omega_i \delta (\nabla_a X^{uiua})}_{\delta H_{\text{ZB2}}} - \frac{1}{2} \omega_i \partial_v \delta X^{uvui} + \underbrace{X^{uivj}\partial_v \delta K_{ij} + X^{uivj} \delta R_{vivj}}_{\delta H_{\text{ZB3}}}\\
            & \qquad + 2(\partial_v \bar K_{ij})\left(4 \pdv{X^{uiuj}}{R_{vkvl}} \delta R_{vkvl} + \pdv{X^{uiuj}}{(\nabla_v V_v)} \delta (\nabla_v V_v) \right) - (\partial_v \bar K_{ij}) \delta X^{uiuj} \biggr)
        \end{split}
    \end{equation}

    We first calculate 
    \begin{equation}
        R_{viuj} = - \partial_v \bar K_{ij} + \frac{1}{2} D_i \omega_j - \frac{1}{4} \omega_i \omega_j
    \end{equation}
    where we've thrown the term $\sim \bar K K$ as $K_{ij} = 0$ on the background, and recall
    \begin{equation}
        \delta R_{vivj} = - \partial_v \delta K_{ij}.
    \end{equation}

    Immediately, we see 
    \begin{equation}
        \delta H_{\text{ZB3}} = X^{uivj}\partial_v \delta K_{ij} + X^{uivj} \delta R_{vivj} = X^{uivj}(\partial_v \delta K_{ij} - \partial_v \delta K_{ij}) = 0.
    \end{equation}

    We then calculate 
    \begin{equation}
        \delta H_{\text{ZB2}} = \frac{1}{2} \omega_i \delta (\nabla_a X^{uiua}) = \frac{1}{2} \omega_i \partial_v \delta X^{uiuv} + \frac{1}{4} \omega_i \omega_j \delta X^{uiuj} + \frac{1}{2} \omega_j D_i \delta X^{ujui},
    \end{equation}
    and 
    \begin{equation}
        \delta H_{\text{ZB1}} = - (\partial_v \bar K_{ij}) \delta X^{uiuj} + \frac{1}{2} (D_i \omega_j)\delta X^{uiuj} - \frac{1}{4} \omega_i \omega_j \delta X^{uiuj}.
    \end{equation}

    We can also use the same argument above to expand 
    \begin{equation}
        \delta X^{uiuj} = 4 \pdv{X^{uiuj}}{R_{vkvl}} \delta R_{vkvl} + \pdv{X^{uiuj}}{(\nabla_v V_v)} \delta (\nabla_v V_v).
    \end{equation}

    Combining these together and using the symmetries of Riemann tensor, we get 
    \begin{equation}
        \delta H_{\text{ZB}} = 2 (\omega_j D_i \delta X^{uiuj} + (D_i \omega_j) \delta X^{uiuj}) = 2 D_i(\omega_j \delta X^{uiuj}).
    \end{equation}

    Finally,
    \begin{equation}
        \boxed{\delta E_{\text{ZB}} = \delta H_{\text{ZB}} + \delta \mathcal{V}_{\text{ZB}} = 2 D_i (\omega_j \delta X^{uiuj}) = D_i \left( \partial_v \delta J_{\text{ZB}}^i \right)}\,.
    \end{equation}

    Here we get the expected codimension-2 total derivative term 
    \begin{equation}
        \partial_v \delta J_{\text{ZB}}^i = 2 \omega_j \delta X^{uiuj}.
    \end{equation}

    Also, this analysis applies trivially to the $f(\text{Riemann})$ case considered in \cite{Wall:2015raa} by setting $V, \nabla V = 0$.

    \subsection{Increasing Entropy}\label{sec:vec-formula}
    We quickly summarise the results so far: 

    \begin{equation}
        \boxed{\delta E_{vv} = \frac{1}{\sqrt{h}}\partial_v^2 \delta \left( \sqrt{h} (\varsigma_{\text{Wald}} + \varsigma_{\text{JKM}}) \right) + D_i \left( \partial_v \delta J^i \right)}
    \end{equation}
    where 
    \begin{equation}
        \varsigma_{\text{Wald}} = 4 \pdv{L_1}{R_{uvuv}}
    \end{equation}
    is the \emph{Wald entropy density},
    \begin{equation}
    \begin{split}
            \varsigma_{\text{JKM}} & = 16 \pdv{L_1}{R_{uiuj}}{R_{vkvl}}\bar K_{ij} K_{kl} - 4 \pdv{L_1}{R_{uiuj}}{(\nabla_v V_v)} \bar K_{ij} V_v\\
            & \quad - 4 \pdv{L_1}{(\nabla_u V_u)}{R_{vivj}} V_u K_{ij} + \pdv{L_1}{(\nabla_u V_u)}{(\nabla_v V_v)} V_u V_v
        \end{split}
    \end{equation}
    is the \emph{dynamical corrections to Wald entropy density}, and 
    \begin{equation}
        \partial_v \delta J^i = \partial_v \delta J_{\text{M}}^i + \partial_v \delta J_{\text{ZB}}^i = 4(\delta (\nabla_a X^{uiua}) + \partial_v \delta X^{uvui}) + 2 \omega_j \delta X^{uiuj} + V^i \delta Z_{v|v} \label{eq:ent-cur}
    \end{equation}
    is the derivative of the \emph{entropy current}.

    Now, using the same definition as in \cite{Wall:2015raa}, we define our increasing entropy for the perturbed black hole on a horizon time-slice $\mathcal{C}(v)$ through
    \begin{equation}
        \boxed{\partial_v^2 \delta S_{\text{inc}} = - 2 \pi \int_{\mathcal{C}(v)} \dd[D-2]{x} \sqrt{h}\, \delta E_{vv}}\,.
    \end{equation}

    Then 
    \begin{equation}
    \begin{split}
            \partial_v^2 \delta S_{\text{inc}} & = - 2 \pi \partial_v^2 \delta \int_{\mathcal{C}(v)} \dd[D-2]{x} \sqrt{h}\, (\varsigma_{\text{Wald}} + \varsigma_{\text{JKM}})
        \end{split}
    \end{equation}
    where we integrated out $D_i J^i$ by assuming $\mathcal{C}(v)$ is compact.

    Finally, we have, to first order in perturbation,
    \begin{equation}
        \boxed{\begin{aligned}
            S_{\text{inc}} & = - 2 \pi \int_{\mathcal{C}(v)} \dd[D-2]{x} \sqrt{h}\, \biggl(4 \pdv{L_1}{R_{uvuv}} + 16 \pdv{L_1}{R_{uiuj}}{R_{vkvl}}\bar K_{ij} K_{kl}\\
            & \qquad - 4 \pdv{L_1}{R_{uiuj}}{(\nabla_v V_v)} \bar K_{ij} V_v - 4 \pdv{L_1}{(\nabla_u V_u)}{R_{vivj}} V_u K_{ij} + \pdv{L_1}{(\nabla_u V_u)}{(\nabla_v V_v)} V_u V_v \biggr).
        \end{aligned}}\label{eq:entropy-1}
    \end{equation}

    In the presence of some \emph{additional} minimal matter sector obeying the null energy condition as discussed in Section \ref{sec:min-cpl}, the perturbation of gravitational sector is sourced by external matter field satisfying NEC, i.e., $\delta E_{vv} = \delta T_{vv} \geq 0$. If we further impose the teleological boundary condition that the perturbation settles down at late time, i.e., $\partial_v S_\text{inc} \to 0$ as $v \to \infty$, the increasing entropy satisfies a linearised second law
    \begin{equation}
        \boxed{\partial_v S_{\text{inc}} \geq 0.}
    \end{equation}

    In the case of a gauge field $A_a$, the coupling between $A_a$ and gravity is through the gauge-invariant curvature $F_{ab}$ (or $\tr (F_{ab} \cdots)$ for non-Abelian cases). However, $F_{ab}$ is anti-symmetric, hence the explicit dependence on $\nabla_u A_u$ and $\nabla_v A_v$ of the Lagrangian vanishes, i.e.,
    \begin{equation}
        \pdv{L}{(\nabla_a A_b)} = \pdv{L}{F_{cd}} \pdv{F_{cd}}{(\nabla_a A_b)} = \pdv{L}{F_{cd}}\/(\delta^{a}_{c} \delta^{b}_{d} - \delta^{a}_{d} \delta^{b}_{c}) \quad \Rightarrow \quad \pdv{L}{(\nabla_u A_u)} = \pdv{L}{(\nabla_v A_v)} = 0,
    \end{equation}
    therefore, for this specific formula (\ref{eq:entropy-1}), the last three terms vanish. For theories with higher order derivatives, its contribution will in general enter the formula, but that would be another story to tell. Details on these are discussed in \cite{Biswas:2022grc}.

    Another curious feature of the above entropy formula \eqref{eq:entropy-1} is that we may factorise the derivatives on $L$ in the JKM correction as 
    \begin{equation}
        S = S_\text{Wald} + \int \dd[D-2]{x} \sqrt{h} \left(4 \bar K_{ij} \pdv{R_{uiuj}} - V_u \pdv{(\nabla_u V_u)} \right) \left(4 K_{kl} \pdv{R_{vkvl}} - V_v \pdv{(\nabla_v V_v)} \right) L_1.
    \end{equation}
    Such factorisability is not yet understood generally and abstractly, and whether it happens at higher order of derivatives relies on more technical and tedious calculations. It would be an interesting future direction to explore, in order to understand the abstract structure of the increasing entropy.

    \section{Discussion and Outlook}\label{sec:discussion}

    Let us summarise our results.  An increasing entropy $S_\text{inc}$ was defined at first order, for all higher curvature gravities.  Because the entropy can be written in covariant phase space language, it is manifestly covariant.  However, it differs from Wald's original proposal $\int \bm Q$ by a correction term proportional to the pre-symplectic potential $\bm \Theta$, ensuring its invariance under JKM ambiguities.  The resulting formula for $S_\text{inc}$ is equivalent (at first order) to the one obtained by stripping off 2 derivatives from $E_{vv}$ in \cite{Wall:2015raa}. 

    Vector fields can also be included with arbitrary covariant couplings.  For an action of the form $I(\text{Riemann}, V, \nabla V$), there are 3 additional terms relative to Dong's entropy formula \cite{Dong:2013qoa} (Section \ref{sec:vec-formula}):
    \begin{equation}
        \Delta S \propto \int_{\mathcal{C}(v)}\!\!\!\! \dd[D-2]{x} \sqrt{h}\, \biggl( \pdv{L_1}{R_{uiuj}}{(\nabla_v V_v)} \bar K_{ij} V_v + \pdv{L_1}{(\nabla_u V_u)}{R_{vivj}} V_u K_{ij} - \frac{1}{4} \pdv{L_1}{(\nabla_u V_u)}{(\nabla_v V_v)} V_u V_v \biggr).
    \end{equation}
    This formula also works for multiple vector fields if you sum over a flavour index. We also found the ($v$-derivative of) entropy current $V^i \delta Z_{v|v}$ associated with the vector (appearing in \eqref{eq:ent-cur}).  We expect many additional correction terms for a higher number of derivatives --- we do not yet have a simple formula that works at all orders, although the algorithm in Section \ref{sec:prob-term} can in principle be used to identify the entropy for metric-vector-scalar theories with an arbitrary number of derivatives.

    The formula for $S_\text{inc}$ given above in \eqref{eq:entropy-1} is time-reversal symmetric, in the sense that it goes to itself under exchanging $u \leftrightarrow v$.  Although this is plausible physically, from an algebraic point of view it is not obvious why it happened, as the procedure for getting it --- stripping two $\partial_v$ derivatives off $E_{vv}$ on the future horizon --- is not manifestly symmetric in $u$ and $v$.  It would be illuminating to find a conceptual reason why this must happen generally (assuming it does), even if we allow more derivatives.

    It would also be interesting to consider cases with non-minimally coupled matter fields with spin higher than 1.\footnote{Aside from the massless graviton, which we are already allowing at linear order.} We will call these fields ``higher-spin fields'' for now. Similar to vector fields, we can use the constraint form $\delta \bm C_k$ to define an entropy which satisfies the linearised second law. This now involves the equation of motion of the higher-spin field in addition to the metric equation of motion, i.e., the second derivative of entropy reads
    \begin{equation}
        \partial_v^2 \delta S \sim - \int \dd[D-2]{x} \sqrt{h} \left(\delta E_{vv} + \lambda \delta (\mathcal{E}_{v}^{\cdots} \Phi_{v\cdots})\right)
    \end{equation}
    where $\Phi$ is the schematic higher-spin field and $\mathcal{E}$ is its equation of motion. $\lambda$ is some theory-dependent constant. The dots are the corresponding index contractions. Here, the additional $\delta (\mathcal{E}_{v}^{\cdots} \Phi_{v\cdots})$ term may not vanish at first order in perturbation. E.g., for a spin-3 field $\Phi_{abc}$, it includes terms such as:
    \begin{equation}
        \delta (\mathcal E\indices{_v^{ab}} \Phi_{v ab})\quad \supset \quad \mathcal E_{vuu} \delta \Phi_{vvv},\, \mathcal E\indices{_{vu}^i} \delta \Phi_{vvi},\, \Phi_{vuu} \delta \mathcal E_{vvv}.
    \end{equation}
    In order to show that the entropy is constant we will also need to assume the higher spin equation of motion is satisfied at both zeroth and first order.  That is, the entropy is constant when the perturbed geometry and matter fields satisfy equations $\delta E = 0$, ${\cal E} = 0$, and $\delta \mathcal E = 0$. If a minimal matter sector obeying null energy condition is present, we can define an entropy that is strictly non-decreasing.  This higher spin case will be discussed in more detail in an upcoming paper \cite{Yan}.

    Besides the generalisation to higher-spin matter, what happens at second order of the perturbation is also worth investigating --- it is the exact order where a non-trivial second law manifests (as at linear order the most general second law really means constancy).  At this order we believe it is not true that every higher curvature gravity theory will obey a second law.  There will be substantive constraints on what theories are allowed.  \cite{Hollands:2022fkn} studied this case using EFT assumptions, and subsequently a non-perturbative second law was proven by Davies and Reall \cite{Davies:2023qaa, Davies:2024fut} using EFT constraints. However, we believe that there may be alternative constraints to EFT assumptions, and there exists special theories where, like GR, an exact second law holds. This seems especially likely in higher spin supergravities (e.g., truncating the superstring action at some order in $\alpha^\prime$), where the positivity properties of supersymmetry may be useful for ensuring a second law.

    Another direction is to relate the entropy formula (\ref{eq:entropy-1}) obtained in this article with the holographic entanglement entropy calculated via the action of a higher curvature bulk theory in the presence of vector fields. In $f(\text{Riemann})$ theory, there is an exact agreement between the holographic entanglement entropy by Dong (Dong entropy) \cite{Dong:2013qoa} with the increasing entropy \cite{Wall:2015raa}. It would be interesting if we could do the calculations following Dong's procedure \cite{Dong:2013qoa} and/or a refined version in the appendix of \cite{Dong:2019piw} for the vector fields, and then see if the agreement still holds.  It would also be useful to find a more general abstract argument concerning why these entropies agree.

    An additional important question concerns the field redefinition covariance of $S$.  Although the formula for $S$ is not obviously field redefinition invariant, we expect that --- to the extent that it is not --- such differences will make no difference to whether the second law is true.  (This might result in a family of increasing entropies, rather than a single formula.)  A simple example concerns the field redefinition of pure gravity (no matter) by $g_{ab} \to g_{ab} + \epsilon R_{ab}$ in which case an $\epsilon R_{ab} R^{ab}$ term appears in the Lagrangian.\footnote{And also an $\epsilon R^2$ term which does not raise additional complications.}  The $S_\text{inc}$ of this term includes not just $R_{uv}$ (which vanishes on a GR solution) but also $K \bar{K}$ which cannot be set to 0 using the GR equation of motion.  However, the correction vanishes for a holographic entropy surface which satisfies the stationary condition $K = 
    \bar{K} = 0$.  On a future black hole horizon ($u = 0$, $v \ne 0$), $\bar{K}$ does not vanish but $K$ is proportional to $\theta$, the increase of entropy $\partial_v S$ in GR.  We believe that, in perturbation theory, the presence of such a term cannot change whether the entropy increases or not. 

    In considering such field redefinitions, one would also need to consider carefully the effects of redefining $g_{vv}$ and $g_{vi}$, as this would change the definition of a null surface, and hence the location of the horizon!  Relatedly, in higher curvature gravity theories it is possible for the characteristic surfaces of gravitational fields to lie outside the horizon, so that the actual causal structure of the theory is given by not by the lightcone of $g_{ab}$, but another surface.  This surface can in general be a Finsler geometry \cite{Javaloyes:2018lex}, as for example in the case of Gauss-Bonnet where different polarisations of the graviton propagate at different maximal speeds in different directions \cite{Camanho:2014apa,Papallo:2015rna}.  This would affect the physical definition of an event horizon, as the boundary of the location from which no physical signals can escape to infinity.  In this work, we have always defined the horizon in the naive way using the metric $g_{ab}$, but if the construction can 
    be shown to be invariant under arbitrary field redefinitions, that would presumably allow this restriction to be lifted, as one could then define $g_{vv}$ and $g_{vi}$ to match the propagation of the fastest field outward on the horizon.

    \section*{Acknowledgements}
    The authors are grateful to Amr Ahmadain, Luca Ciambelli, Prateksh Dhivakar, Laurent Freidel, Ted Jacobson, Nilay Kundu, Zhihan Liu, Prahar Mitra, Harvey Reall, Ronak Soni, Manus Visser, Bob Wald, Diandian Wang, Zi-Yue Wang and Houwen Wu for helpful discussions. This work was supported in part by AFOSR grant FA9550-19-1-0260 ``Tensor Networks and Holographic Spacetime'', STFC grant ST/P000681/1 ``Particles, Fields and Extended Objects'', an Isaac Newton Trust Early Career grant, and NSF PHY-2309135 to the Kavli Institute for Theoretical Physics (KITP). ZY is also supported by an Internal Graduate Studentship of Trinity College, Cambridge.

    \appendix
    \section{Zero-Boost Terms}
    \label{app:zbt}

    Here, we briefly review the treatment of zero-boost terms in \cite{Hollands:2022fkn}. Consider the schematic structure of the null-null component of equation of motion $E_{vv}$ at first-order in perturbation
    \begin{equation}
        \delta E_{vv} = \frac{1}{\sqrt{h}} \partial_v^2 \left(\sqrt{h} \sum_I \sum_{k\geq 0} A_{(-k)}^I \delta B_{(k)}^I\right) + D_i \tilde J^i
    \end{equation}
    where we assume that the inauspicious terms discussed in Section \ref{sec:prob-term} are already resolved. For terms with weight $k > 0$, we can see they are manifestly exact in $\delta$ at first-order 
    \begin{equation}
        A^I_{(-k)} \delta B^I_{(k)} = A^I_{(-k)} \delta B^I_{(k)} + B^I_{(k)} \delta A^I_{(-k)}  = \delta \left(A^I_{(-k)} B^I_{(k)}\right), \quad k>0
    \end{equation}
    by quoting Claim \ref{clm:2}. The problem with zero-boost terms are that they are not manifestly exact in $\delta$ at first-order, i.e., 
    \begin{equation}
        \sum_I A^I_{(0)} \delta B^I_{(0)} \overset{?}{=} \delta (\cdots)
    \end{equation}
    making it difficult to render the integral 
    \begin{equation}
        \int \dd[D-2]{x} \sqrt{h}\, \delta E_{vv} \overset{?}{=} \delta (\cdots)
    \end{equation}
    as a manifestly exact variation (of derivative of the entropy, as we saw earlier). To resolve the above concern, \cite{Wall:2015raa} quotes the physical version of the first law to claim that the zero-boost term corresponds to the Iyer-Wald entropy defined in \cite{Iyer:1994ys}. However, a more local formula was not given until the appearance of \cite{Hollands:2022fkn}.

    In \cite{Hollands:2022fkn}, this is treated by invoking the powerful covariant phase space formalism. Start with identity (\ref{eq:magic-id}) and integrate it on the horizon $\mathcal{H}$ from the bifurcation surface $\mathcal{C}(0)$ to $\mathcal{C}(\infty)$, we get 
    \begin{equation}
        \int_{\frac{1}{2}\mathcal{H}} \delta \bm C_{\xi} = \left(\int_{\mathcal{C}(+\infty)} - \int_{\mathcal{C}(0)} \right) \left(\delta \bm Q_\xi - \iota_\xi \bm \Theta[\delta g, \delta V]\right) \label{eq:zbt-magic-id}
    \end{equation}
    As this is an off-shell identity, we are free to make the assumption that the perturbation is compactly supported around the bifurcation surface $v=0$. We first investigate the LHS of the equation (\ref{eq:zbt-magic-id}):
    \begin{equation}
        \begin{split}
            \text{LHS of (\ref{eq:zbt-magic-id})} & = - \int_{v=0}^{+\infty} \dd{v} \int_{\mathcal{C}(v)} \dd[D-2]{x} \sqrt{h}\, v\, \delta E_{vv}\\
            & = - \underbrace{\left[v \int_{\mathcal{C}(v)} \dd[D-2]{x} \sqrt{h} \cdots\right]_{v=0}^{+\infty}}_{0} + \left[\int_{\mathcal{C}(v)} \dd[D-2]{x} \sqrt{h}\, \sum_I \sum_{k\geq 0} A^I_{(-k)}\delta B^I_{(k)}\right]_{v=0}^{\infty}\\
            & = -\int_{\mathcal{C}(0)} \dd[D-2]{x} \sqrt{h}\, \sum_I \sum_{k\geq 0} A^I_{(-k)}\delta B^I_{(k)}.
        \end{split}
    \end{equation}
    Note that we used the assumption that the time-slices of the horizon are compact. Here, we recall the result in the proof of Claim \ref{clm:2} that the background quantities scale as 
    \begin{equation}
        A^I_{(-k)} \sim v^k, \quad k \geq 0. \label{eq:A-scale}
    \end{equation}
    So, when evaluated at $v=0$, only the zero-boost term is left, so the whole identity evaluates as
    \begin{equation}
        \int_{\mathcal{C}(0)} \dd[D-2]{x} \sqrt{h}\, \sum_I A^I_{(0)} \delta B^I_{(0)} = \int_{\mathcal{C}(0)} \delta \bm Q_\xi = \delta S_{\text{IW}}\big|_{v=0}
    \end{equation}
    where the RHS is identified as the Iyer-Wald entropy which contains all the boost-invariant part (by the same argument as (\ref{eq:A-scale}), etc.) of the perturbed Wald entropy according to \cite{Iyer:1994ys}. 

    One might wonder if this identity only holds at $v=0$. The answer is that it should hold for any $v$ by a clever argument as follows. Firstly, the background quantities $A^I_{(0)}$ are constant with respect to $v$ along the horizon by the Killing weight argument. Secondly, for the perturbation part, we can just translate the compactly supported perturbations via a diffeomorphism $T_a: v \mapsto v + a$ to obtain another perturbation $\delta B^I_{(0)} \to T^*_a \delta B^I_{(0)}$ via pushforward/pullback which is compactly supported around $v = a$. Since the identity is off-shell, it should hold for any perturbation. We then obtain 
    \begin{equation}
        \int_{\mathcal{C}(a)} \dd[D-2]{x} \sqrt{h}\, \sum_I A^I_{(0)} T_a^*\delta B^I_{(0)} = \int_{\mathcal{C}(a)} T_a^*\delta \bm Q_\xi = \delta S_{\text{IW}}\big|_{v=a}
    \end{equation}
    where we identified $\mathcal{C}(a) = T_a^* \mathcal{C}(0)$. Locally, this suggest that 
    \begin{equation}
        \sum_I A^I_{(0)} \delta B^I_{(0)} = \delta \varsigma_\text{IW} + D_i J_{\text{ZB}}^i
    \end{equation}
    where $\varsigma_\text{IW}$ is the density of the Iyer-Wald entropy, and the last term is a codimension-2 total derivative, which integrates out on compact $\mathcal{C}(v)$. The form of the current $J^i_\text{ZB}$ can be obtained via a similar analysis as in Section \ref{sec:prob-term}.

    \section{Some Details of Calculations in Section \ref{sec:example1}}\label{app:calc}
    \subsection{Justification of Equation (\ref{eq:Y_ab})}
    We start from the Lie derivative of the Lagrangian with respect to some diffeomorphism generator $\xi$
    \begin{equation}
        \pounds_\xi L_1 = X^{abcd} \pounds_\xi R_{abcd} + Y_{ab} \pounds_\xi g^{ab} + Z^a \pounds_\xi V_a + Z^{a|b} \pounds_\xi (\nabla_b V_a). \label{eq:L1-Lie}
    \end{equation}
    The LHS evaluates to 
    \begin{equation}
        \text{LHS of (\ref{eq:L1-Lie})} = \xi^e \nabla_e L_1 = \xi^e \left(X^{abcd} \nabla_e R_{abcd} + Z^a \nabla_e V_a + Z^{a|b} \nabla_e (\nabla_b V_a)\right)
    \end{equation}
    while the RHS is 
    \begin{equation}
    \begin{split}
            \text{RHS of (\ref{eq:L1-Lie})} & = X^{abcd} \pounds_\xi R_{abcd} + Y_{ab} \pounds_\xi g^{ab} + Z^a \pounds_\xi V_a + Z^{a|b} \pounds_\xi (\nabla_b V_a)\\
            & = X^{abcd}\/(\xi^e \nabla_e R _{abcd} + 4 (\nabla_a \xi^e) R _{ebcd}) - 2 Y_{ab} \nabla^a \xi^b + Z^a \left(\xi^e \nabla_e V_a + V_e \nabla_a \xi^e\right)\\
            & \quad + Z^{a|b} (\xi^e \nabla_e (\nabla_b V_a) + (\nabla_b V_e) \nabla_a \xi^e + (\nabla_e V_a) \nabla_b \xi^e ).\\
        \end{split}
    \end{equation}
    
    Cancelling the common terms on both sides we have 
    \begin{equation}
        (\nabla^a \xi^b)(4 X\indices{_a^{cde}}R_{bcde} - 2 Y_{ab} + Z_a V_b + Z_{a|c}\nabla^c V_b + Z_{c|a}\nabla_b V^c) = 0
    \end{equation}
    which should hold for all $\xi$. Hence (note that $Y_{ab}$ only captures the totally symmetrised part)
    \begin{equation}
        Y_{ab} = 2 X\indices{_{(a}^{cde}}R_{b)cde} + \frac{1}{2} (Z_{(a} V_{b)} + Z_{(a|c}\nabla^c V_{b)} + Z_{c|(a}\nabla_{b)} V^c).
    \end{equation}

    This should be an identity as it is essentially a theory-dependent functional structure which is independent of variations in the configuration space.
    
    \subsection{Omitted Derivations in Section \ref{sec:example1}}

    We present the omitted details for the calculations in Section \ref{sec:example1}. We try to make the cancellations clear. These would fill up the gaps in equations \eqref{eq:H-2}, \eqref{eq:H-1-1}, \eqref{eq:dv-delta-D-X}, \eqref{eq:V-1}, \eqref{eq:V-2-2}, and \eqref{eq:V-2-3}.
   
    \begin{align*}
        \delta H^{(2)} & = 2 \delta (X\indices{_v ^{cde}} R_{vcde})\\
        & = 4 (R_{viuj} \delta X\indices{_v^{iuj}} + X\indices{_v^{ivj}} \delta R_{vivj})+ \underbrace{4 \delta (X\indices{_v^{iuv}} R_{viuv}) + 2 \delta (X\indices{_v^{ijk}} R_{vijk})}_{0}\\
        \delta H^{(1,1)} & = 4 \delta (\partial_v (\nabla_a X^{uvua}))\\
        & = 4 \biggl(\partial_v^2 \delta X_{uvuv} - \frac{1}{2} \omega_i \partial_v \delta X^{uvui} +  D_i (\partial_v \delta X^{uvui})\\
        & \qquad +  X^{uivj} \partial_v \delta K_{ij} - \partial_v (\bar K_{ij}) \delta X^{uiuj} -  \bar K_{ij} \partial_v \delta X^{uiuj} +  X^{uvuv} \partial_v \delta K \biggr)\\
        & = 4 \biggl(\frac{1}{\sqrt{h}}\partial_v^2 \delta \left(\sqrt{h}\, X_{uvuv}\right) - \frac{1}{2} \omega_i \partial_v \delta X^{uvui} +  D_i (\partial_v \delta X^{uvui})\\
        & \qquad +  X^{uivj} \partial_v \delta K_{ij} - \partial_v (\bar K_{ij}) \delta X^{uiuj} -  \bar K_{ij} \partial_v \delta X^{uiuj} \biggr)\\
        \partial_v (\delta D_i X^{uvui}) & = \partial_v (\partial_i \delta X^{uvui} + \delta (\conn{i}{i}{j} X^{uvuj}))\\
        & = \partial_i \partial_v \delta X^{uvui} + \underbrace{\partial_v(\delta \conn{i}{i}{j})X^{uvuj} + \partial_v(\conn{i}{i}{j})\delta X^{uvuj} + \delta \conn{i}{i}{j} \partial_v X^{uvuj}}_{0} + \conn{i}{i}{j} \partial_v \delta X^{uvuj}\\
        & = D_i(\partial_v \delta X^{uvui})\\
        \delta \mathcal{V}^{(1)} & = Z_{u|v} \delta (\nabla_v V^u) + (\nabla_v V^v) \delta Z_{v|v} + \underbrace{\delta (Z_{i|v} \nabla_v V^i)}_{0}\\
        \delta \mathcal{V}^{(2,2)} & = V_u \delta (\nabla_v Z_{v|v}) + (\nabla_v V^v) \delta Z_{v|v} - \underbrace{\delta (\nabla_v(Z^{v|u}) V_v)}_{0} - Z_{u|v} \delta (\nabla_v V^u)\\
        & = V_u \partial_v(\delta Z_{v|v}) + (\nabla_v V^v) \delta Z_{v|v} - Z_{u|v} \delta (\nabla_v V^u)\\
        \delta \mathcal{V}^{(2,3)} & = D_i (V^i \delta Z_{v|v} - \underbrace{\delta (Z\indices{^i_{|v}} V_v)}_{0}) - \underbrace{\delta (K\indices{_i^j}(Z_{j|v}V^i - Z\indices{^i_{|j}} V_v))}_{0} + \frac{1}{2} \omega_i (V^i \delta Z_{v|v} - \underbrace{\delta (Z\indices{^i_{|v}} V_v)}_{0})\\
        & \qquad - \underbrace{\delta (K\indices{_i^j} (Z_{v|j}V^i - Z\indices{^i_{|v}} V_j))}_{0} + \frac{1}{2} \omega_i (V^i \delta Z_{v|v} - \underbrace{\delta (Z\indices{^i_{|v}} V_v)}_{0})\\
        & \qquad + \underbrace{\delta (K (Z_{v|v}V^v - Z\indices{^v_{|v}} V_v))}_{0} + \underbrace{\delta (\bar K (Z_{v|v}V^u - Z\indices{^u_{|v}} V_v))}_{0}\\
        & = D_i (V^i \delta Z_{v|v} ) + \omega_i V^i \delta Z_{v|v}.
    \end{align*}

    \bibliographystyle{JHEP}
    \bibliography{ppr1}
    
\end{document}